\documentclass{article}[12pt]

\usepackage{graphicx}
\usepackage{amsfonts,amsmath,amsthm}
\usepackage[longnamesfirst]{natbib}
\usepackage[page,title]{appendix}
\usepackage{amssymb}
\usepackage{booktabs}
\usepackage{multirow}
\usepackage{pdflscape}
\usepackage{rotating}
\usepackage{setspace}
\usepackage{geometry}
\usepackage{url}
\usepackage{natbib}

\usepackage{etoolbox}
\patchcmd{\appendices}{\quad}{: }{}{}

\newtheorem{theorem}{Theorem}[section]
\newtheorem{proposition}{Proposition}[section]

\newtheorem{definition}{Definition}[section]
\newtheorem{example}{Example}

\newtheorem{remark}{Remark}[section]
\newtheorem{assumption}{Assumption}[section]

\newtheorem{condition}{Condition}[section]

\begin{document}
\title{Moment Inequalities in the Context of Simulated and Predicted Variables}
\date{April 7, 2018}

\author{Hiroaki Kaido\thanks{Email: hkaido@bu.edu. Financial support from NSF grant SES-1357643 is gratefully acknowledged.} \\Department of Economics\\Boston University \and Jiaxuan Li\thanks{Email: lijiaxuan0529@gmail.com} \\Amazon.com  \and Marc Rysman\thanks{Email: mrysman@bu.edu.} \\Department of Economics\\Boston University }

 \maketitle

\begin{abstract}
This paper explores the effects of simulated moments on the performance of
inference methods based on moment inequalities. Commonly used confidence
sets for parameters are level sets of criterion functions whose boundary
points may depend on sample moments in an irregular manner. Due to this feature, simulation errors can affect the performance of
inference in non-standard ways. In particular, a (first-order) bias due to
the simulation errors may remain in the estimated boundary of the
confidence set. We demonstrate, through Monte Carlo experiments, that
simulation errors can significantly reduce the coverage probabilities of
confidence sets in small samples. The size distortion is particularly
severe when the number of inequality restrictions is large. These results
highlight the danger of ignoring the sampling variations due to the
simulation errors in moment inequality models. Similar issues arise when using predicted variables in moment inequalities models.  We propose a method for properly correcting for these variations based on
regularizing the intersection of moments in parameter space, and we show that our proposed method performs well theoretically and in practice.
\end{abstract}

\vspace{0.4in}


\clearpage
\onehalfspacing
\section{Introduction}
Recently, estimating with moment inequalities rather than traditional equalities has proven increasingly popular \citep{Tamer2010}.   Even in contexts when a model is difficult or even impossible to compute exactly, economic theory may provide inequalities that are amenable to estimation.  Thus, many of the examples in which moment inequalities are attractive are also examples in which computing the model is complex.  This model complexity also leads the author to use simulation techniques, often to address the calculation of integrals over complex objects. However, the role of simulation in moment inequality estimation has been largely unexplored.

In this paper, we argue that using simulation or predicted values in the context of moment inequalities has important implications both for determining identified sets and for inference.  We show that simulated moment inequality estimators suffer from small sample bias, which is a result of the irregular nature of the estimator.  This irregularity is introduced by the following mechanism.  Rather than finding optimal parameters through some kind of extremum estimator, moment inequality estimators typically build the estimated identified set or confidence region through level set computations, such as in \cite{ChernozhukovHanTamer2007} and \cite{AndrewsJia2012}.    The level set is defined by the intersection of moments, and at these intersections, the objective function may depend on the underlying moments in an irregular way.  That is, at such intersections, the distribution of the objective function can be particularly sensitive to perturbations to the underlying moments. Since these intersections often define the maximum or minimum of the confidence interval in a particular direction, they determine the range of parameters in the confidence interval, and thus the intersections are often of greatest interest. Whereas approximation error (for instance, due to simulation) in moment equality estimators for small samples is of second order importance, we show that irregular objective functions promote approximation error in small samples to first order importance.

In this paper, we first describe formally the phenomena that we are interested in. A starting place is the well-known result for moment equalities that estimation that requires simulation is consistent, even for a fixed number of draws \citep{McFadden:1989kx,PakesPollard1989}.
 Moment inequality estimators are similar to moment equalities because they take means over simulated or predicted variables before transforming them into an objective function.
Simulated moment  functions  are consistent (and asymptotically unbiased) for their population counterparts $E[m(X,\cdot)]$ for a finite number of draws.
 This ensures that level-set estimators of identified sets \citep{ChernozhukovHanTamer2007} are also consistent in the Hausdorff distance.

We show that the similarity breaks down when it comes to inference in finite samples, due to the irregular feature of the objective function.  We explore this phenomenon in two Monte Carlo experiments.  The first one focuses on estimation of the point of intersection of multiple simulated moment inequalities, and is meant to maximize the scope of the problem we discuss.  The second is a more realistic treatment: We study an entry game in the spirit of \cite{CilibertoTamer2009}. Moment inequalities are generated by the lower and upper bounds of entry probabilities conditional on covariates that are classified into a finite number of bins.  In our Monte Carlo experiments, we find that the coverage probabilities of the confidence regions are distorted severely when there are bins with a limited number of observations and a small number of simulation draws is used. The presence of such bins is common in empirical applications. Therefore, inference methods that properly account for the effects of simulation are needed.

We propose a new solution to the problems associated with the irregularity of the objective function, including those due to simulated moments.
The key is  that the commonly used objective functions involve \emph{smoothable convex functions} \citep{Beck:2012xy}.
 Our method ``regularizes" or smoothes the objective function, and we show that this method leads to a straightforward bias correction method.  While the idea of regularization appears in some other contexts, such as  \cite{HaileTamer2003}, \cite{Chernozhukov:2015aa} and \cite{Masten:2017aa}, we formally show that this approach has uniform validity in the context of inference with simulated variables. Our regularization method is based on the class of $\mu$-smooth approximations studied in the non-smooth optimization literature \citep{Nesterov:2005qv,Beck:2012xy}. We provide conditions on the choice of approximating functions and regularization parameters that ensure the uniform validity of an inference procedure that combines the proposed regularization scheme with a straightforward bootstrap resampling. In addition to being attractive theoretically, we show that this approach performs well in practice, even relative to techniques that implement bias correction via the adjustment of the critical value such as \cite{Andrews:2010jk} and \cite{Chernozhukov:2013fj}. Specifically, our Monte Carlo experiments show that the proposed method controls the size well and is often less conservative than the existing methods.

Before moving on, we want to stress that simulation is common in many well-known applications of moment inequalities.  For instance, \cite{HaileTamer2003} simulate bids in order to place bounds on the value of participants in an auction.  \cite{CilibertoTamer2009} simulate an entry game to determine upper and lower bounds for the probability of a firm entering under different equilibrium selection mechanisms.  \cite{Ho2009} uses inequalities to study network formation between hospitals and insurers, and uses predicted profit from a network as an explanatory variable for the firm's choices.  The profit function is based on a random coefficient logit demand function as in \cite{BerryLevPakes1995}, which involves simulation.  \cite{Eizenberg2011} studies firms choosing product characteristics, and also relies on simulated demand predictions drawn from a \cite{BerryLevPakes1995} demand system.  \cite{KawaiWatanabe2013} simulates market effects in a model of strategic voting that uses moment inequalities to address unobserved beliefs about other voters.  Although not strictly moment inequalities, the method of \cite{BajariBenkardLevin} uses simulation in the context of a minimum distance estimator based on inequalities to study dynamic oligopoly games.  Well-known applications are \cite{Ryan2012} and \cite{FowlieReguantRyan2015}.

While our paper focuses on simulation as a source of small-sample bias, the same problem is introduced by using predicted values from some prior estimation stage.  For example, \cite{Holmes2011} studies the diffusion of Walmart using moment inequalities, and \cite{HoudeNewberrySeim2017} take a similar approach to study the locations of Amazon's fulfillment centers.  Both papers use profits or revenues as explanatory variables in their moment inequalities, where revenues and profits are constructed from estimated models.  Although neither of these papers use simulation in any stage of their estimation, the fact that there is estimation error associated with these variables brings up similar issues to the approximation error introduced by simulation.  

Our paper follows in a long line of research on problems with using simulation and prediction in estimation procedures.  For instance, \cite{Hausman:1983aa} terms using predicted or simulated values in non-linear estimation procedures the ``Forbidden Regression,'' and \cite{GourierouxMontfort1996} shows that simulated Maximum Likelihood is inconsistent for any fixed number of samples due to the non-linearity of maximum likelihood estimator.  Our result is in fact not due to non-linearity, as moment inequality estimators are not inherently non-linear.  Indeed, we show consistency even for a fixed number of samples. Rather, our result emphasizes the irregularity of the level set estimator in a moment inequalities context, which makes the confidence interval irregular at important points.  Interestingly, \cite{BajariBenkardLevin} show in Table 5 that their inequalities estimator may be inferior to an estimator based on moment equalities \citep[as in][]{PakesOstrovskyBerry2007} if one can be implemented.  They ascribe this to the non-linearity of the second-stage estimator.\footnote{For example, on page 1362, \cite{BajariBenkardLevin} write: ``The results above suggest that the inequality estimator may exhibit bias in small samples. This bias arises because the second-stage objective function is nonlinear in the first-stage estimates."}  Our paper provides an alternative explanation, which is the irregularity of the level-set estimator in the context of first-stage simulation.


\section{Setup}\label{sec:setup}
\subsection{Simulated moments and motivating examples}
Let $X_i\in\mathcal X\subset \mathbb R^{d_X}$ be a random vector, $\theta\in\Theta\subset\mathbb R^{d_\theta}$ be a structural parameter and $m:\mathcal X\times\Theta\to\mathbb R^{J}$ be a function known up to the parameter. Consider the (unconditional) moment inequality restrictions:
\begin{align}
	E_P[m_j(X_i,\theta)]\le 0,~j=1,\cdots, J.\label{eq:momineq}
\end{align}
We call the set of parameter values satisfying these restrictions an \emph{identified set} and denote it by $\Theta_I$.
In models where $m_j$ is difficult to evaluate analytically, simulation methods are often employed to obtain its approximation. Throughout, we consider the setting where $m_j$ can be written
\begin{align}
	m_j(x,\theta)=\int M_j(x,u,\theta) dP(u|x),~j=1,\cdots, J,
\end{align}
for some known function $M_j$ and a conditional  distribution $P(\cdot|x)$, which may also depend on the parameter. This allows one to draw simulated samples $u_r,r=1,\cdots,R$ from the conditional distribution and approximate $m_j$ by a simulation counterpart $R^{-1}\sum_{r=1}^RM_j(x,u_r,\theta)$.  We use the subscript $R$ to denote statistics constructed from simulated samples.
This setting parallels the classical  method of simulated moments (MSM) \citep{McFadden:1989kx,PakesPollard1989} except that the moment conditions in \eqref{eq:momineq} involve inequality restrictions.

For making inference for the structural parameter $\theta$ or its identified set (the set of $\theta$s satisfying \eqref{eq:momineq}),
 level-sets of criterion functions are commonly used \citep{ChernozhukovHanTamer2007,Andrews:2010jk}.
Following the literature, we consider set estimators and  confidence regions of the form:
\begin{align}
	\mathcal C_{n,R}=\{\theta\in\Theta:T_{n,R}(\theta)\le c_{n,R}(\theta)\},\label{eq:conf_region}
\end{align}
where  $T_{n,R}$ is a test statistic (properly scaled sample criterion function), and $c_{n,R}$ is a possibly data-dependent critical value. The variable $n$ is the number of observations in our sample, and the subscript {\em n} denotes statistics constructed from that sample. Throughout, we consider criterion functions that can be written as
\begin{align}
	T_{n,R}(\theta)=S(\sqrt n\bar m_{n,R}(\theta),\hat\Sigma_{n,R}(\theta)),\label{eq:test_stat}
\end{align}
for some \emph{index function} $S:\mathbb R^J\times\mathbb R^{J\cdot J}\to\mathbb R$, which aggregates the vector of sample moments $\bar m_{n,R}(\theta)\equiv (nR)^{-1}\sum_{i=1}^n\sum_{r=1}^RM_j(x,u_r,\theta)$ normalized by an estimator $\hat\Sigma_{n,R}$ of the asymptotic covariance matrix. Examples include $S(m,\Sigma)=\max_{j=1}^J\Sigma^{-1/2}_{jj}m_j$ and
$S(m,\Sigma)=\sum_{j=1}^J(\Sigma^{-1/2}_{jj}m_j)_+^2.$

A key observation is that the level sets commonly used in the literature may depend on the underlying moments and hence simulation errors in an irregular manner.
 We illustrate this point using simplifications of well-known examples in the literature.

\begin{example}[Intersection bounds]\label{ex:twomean}\rm
Let $\theta$ be a scalar parameter, let $X_1,X_2\in \mathbb R$ be random variables, and let moment inequalities be given by
\begin{align}
	&\theta-E_P[1\{u_1< X_1\}]\le 0\label{eq:twomean1}\\
	&\theta-E_P[1\{u_2< X_2\}]\le 0,\label{eq:twomean2}
\end{align}
where $(u_1,u_2)$ follows a known distribution $P(\cdot|x)$. That is, the upper bound for $\theta$ is the minimum of two expectations of draws from separate Bernoulli distributions.  While these moment restrictions may appear overly simple, they capture some of the common features shared by empirical examples. These include (i) key parameters are restricted through an intersection of multiple bounds; and (ii) simple frequency simulators can be used to approximate individuals' or firms' choice probabilities represented by the expectation of the indicator functions.

For each $j$,  let $\bar m_{j,n,R}(\theta)=\theta-(nR)^{-1}\sum_{i=1}^n\sum_{r=1}^R1\{u_{j,i,r}<X_{j,i}\}$.
Taking $S(m,\Sigma)=\max_{j=1,2}\{m_j\}$,
one may then construct a confidence interval for $\theta$ with level $1-\alpha$ as follows:
\begin{multline}
	\mathcal C_{n,R}^{\text{Sim}}=\{\theta\in\mathbb R:\sqrt n\max\{\bar m_{1,n,R}(\theta),\bar m_{2,n,R}(\theta)\}\le c\}\\
	=\Big(-\infty,\min\Big\{\frac{1}{nR}\sum_{i=1}^n\sum_{r=1}^R1\{u_{1,i,r}<X_{1,i}\}+c/\sqrt n,~\frac{1}{nR}\sum_{i=1}^n\sum_{r=1}^R1\{u_{2,i,r}<X_{2,i}\}+c/\sqrt n\Big\}\Big],\label{eq:twomean3}
\end{multline}
where  $c$ is a suitable critical value.\footnote{For example a critical value $c$ based on the least favorable configuration where the two constraints bind, i.e. $E_P[1\{u_1< X_1\}]=E_P[1\{u_2< X_2\}]$  solves $P(\max\{W_1,W_2\}\le c)=1-\alpha,$ where
$W=(W_1,W_2)'$ is the distributional limit of $\sqrt n(\bar m_{n,R}-E_P[\bar m_{n,R}])$. A refined critical value based on moment selections can also be used.}  The simulated variables $\{(u_{1,i,r},u_{2,i,r}),r=1,\cdots,R\}$ are drawn from $P(\cdot|X_i)$ for each $i$. As shown in \eqref{eq:twomean3}, the right end point of the confidence interval is given by the minimum of the sample moments (shifted by the critical value).
\end{example}

The next example is an entry game based on \cite{BresnahanReiss1991,Berry1992,Tamer2003aRES,CilibertoTamer2009}.
\begin{example}[Entry game]\label{ex:entry}\rm
	Consider a  binary-response static  game of complete information with two players.
	For each player $j,$ let $Y_j\in\{0,1\}$, $Z_j\in\mathbb R^{d_\beta}$, and $u_j\in\mathbb R$ denote $j'$s binary action, observed and unobserved characteristics respectively.
For each $j,$ let $(\beta_j,\Delta_j)\in\mathbb R^{d_\beta+d_\Delta}$ denote a parameter vector.	The players' payoffs are summarized as follows.
\begin{equation*}
\begin{tabular}{ccc}
& $Y_2=0$ & $Y_2=1$ \\ \cline{2-3}
$Y_1=0$ & \multicolumn{1}{|c}{$0,0$} & \multicolumn{1}{|c|}{$0,Z_2'\beta_1+u_{2}$} \\
\cline{2-3}
$Y_1=1$ & \multicolumn{1}{|c}{$Z_1'\beta_1+u_{1},0$} & \multicolumn{1}{|c|}{$Z_1'\beta_1+u_{1}+\Delta_1,Z_2'\beta_2+u_{2}+\Delta_2$} \\ \cline{2-3}
\end{tabular}%
\end{equation*}%

Suppose that any outcome $(Y_1,Y_2)$ observed by the econometrician is a pure strategy Nash equilibrium and that the opponent's entry $Y_{-j}=-1$ negatively affects a player's payoff, i.e. $\Delta_j<0,$ for $j=1,2$. Then, without further assumptions, the model restricts the conditional probabilities of outcomes as follows:
\begin{align}
P((0,0)|Z) &=	P(u_{1}\le -Z_1'\beta_1,~ u_{2}\le -Z_2'\beta|Z)\label{eq:entry1}\\
P((1,1)|Z) &=	P(u_{1}> -Z_1'
\beta-\Delta_1,~ u_{2}> -Z_2'\beta_2-\Delta_2 |Z)\label{eq:entry2}\\
P((0,1)|Z) &\le P(u_1\le -Z_1'\beta_1-\Delta_1,~u_2> -Z_2'\beta_2|Z)\label{eq:entry3}\\
P((0,1)|Z) &\ge P(u_1\le -Z_1'\beta_1-\Delta_1,~u_2>-Z_2'\beta_2-\Delta_2|Z)\label{eq:entry4}\\
&\qquad\qquad+ P(u_1\le -Z_1'\beta_1,~-Z_2'\beta_2\le u_2\le -Z_2'\beta_2-\Delta_2|Z),\notag
\end{align}
where  $u$ follows a conditional distribution $P(\cdot|Z)$ specified by the researcher, e.g. mean zero bivariate normal with correlation $\rho$. The inequality restrictions \eqref{eq:entry3}-\eqref{eq:entry4} arise because the model predicts multiple equilibria for some values of exogenous variables, while an equilibrium selection mechanism is left unspecified \citep{Tamer2003aRES}.
Suppose for simplicity that $Z=(Z_{1}',Z_2')'$ has a finite support, and its distribution $P_Z(z)$ is known.
The right hand side of \eqref{eq:entry1}-\eqref{eq:entry4} can be approximated by simulators. For example,  the probability $P(u_{1}\le -Z_{1,i}'\beta_1,~ u_{2}\le -Z_{2,i}'\beta_2|Z=z)$ can be approximated by its analog $(nR)^{-1}\sum_{i=1}^{n}\sum_{r=1}^R1\{u_{1,i,r}\le -z_{1,i}'\beta_1,u_{2,i,r}\le -z_{2,i}'\beta_2,Z_i=z\}/P_Z(z)$, where for each $i$,  a sample of simulated payoff shifters $(u_{1,i,r},u_{2,i,r}),r=1,\cdots,R$ are drawn from the conditional distribution $P(\cdot|Z=z)$.
It is straightforward to rewrite the restrictions in \eqref{eq:entry3}-\eqref{eq:entry4} as  unconditional moment inequalities as in \eqref{eq:momineq} for a suitable moment function $m$  with $X_i=(Y_i,Z_i)$.
\end{example}

Whereas \cite{CilibertoTamer2009} places inequalities on the probability of equilibrium outcomes, our next example follows the approach of \cite{PPHI2008} to generate inequalities directly from agent utility functions and revealed preference.  This approach has been utilized to study strategic environments  such as product introductions \citep{Eizenberg2011} and network formation \citep{Ho2009}.  In both of these binary choice examples, the researchers estimate variable profits in a pre-stage and use the moment inequalities to estimate the fixed cost associated with a positive choice.\footnote{Similar examples are \cite{Nosko:2014aa}, \cite{Wollmann:2014aa}, \cite{CrawfordYurukoglu2012}, and \cite{Gowrisankaran:2014aa}.}
  To the extent that variable profits are estimated with some error, this approach introduces analogous problems to the ones we highlight in the context of simulation.  The problem is particularly clear if the variable profits are based on simulation estimators.  That is the case for \cite{Eizenberg2011} and \cite{Ho2009}, which utilize a simulated demand system \citep[i.e.][]{BerryLevPakes1995} in the pre-stage.  We provide an example here, based on \cite{Eizenberg2011}:

\begin{example}[Product introductions]\label{ex:product} \rm
Consider a  binary-response static  game of complete information with two players.
	For each player $j$, let $Y_j\in\{0,1\}$ denote player $j$'s action, let $\pi_j(Y_{-j})$ denote the profits to $j$ from the choice of $Y_j=1$, conditional on the choice of the other firm $Y_{-j}$.  Let $F_j$ be the fixed cost associated with $Y_j=1$, so the payoff to adoption is $\pi_j(Y_{-j})-F_j$.  Let $F_j=F+\zeta_j$, where $\zeta_j$ is observed by the firm and not the researcher, and $E[\zeta_j]=0$.    The firms play a Nash Equilibrium.  We take $\pi_j(\cdot)$ as observed and our goal is to estimate $F$.

If firm $j$ chooses $Y_j=1$, revealed preference implies that $\pi_j(Y_{-j})\geq F_j$.  And similarly, observing $Y_j=0$ implies $\pi(Y_{-j})<F_j$.  We further impose finite bounds on $F_j$, denoted $\overline{F}$ and $\underline{F}$, so we have upper and lower bounds for $F$ that we apply to every observation in the data.  Note that it is common to interact moment inequalities with instrument matrices, and functions of these instruments, to obtain more moments.  We do not explore that here.  Thus, we have the following moment inequalities for $F$:

\[ E_P\left[Y_j\underline{F} + (1-Y_j)\pi_j\left(Y_{-j}\right)\right] \leq F \leq E_P\left[Y_j \pi_j(Y_{-j}) + (1-Y_j)\overline{F}\right],~j=1,2.\]
If $\pi(\cdot)$, $\underline{F}$ and $\overline{F}$ are observed, it is straightforward to construct the sample analog of these inequalities.  However, in practice, these are rarely observed, especially profits for $Y_j=0$.  \cite{Eizenberg2011} estimates a structural demand system and pricing game in a pre-stage in order to construct $\pi(\cdot)$ and how it varies with $Y_{-j}$.  Central to the paper is the use of simulation to construct $\underline{F}$ and $\overline{F}$.  Thus, all of the explanatory variables are approximations and are subject to prediction and simulation error.
\end{example}

We start with an observation that, similar to moment equalities, level-set estimators that use simulation are consistent even for a fixed number of draws.  This similarity arises because  estimators of the moments take means over simulated  variables before transforming them into an objective function and  level-set estimators depend on the estimated moments in a continuous way.
For each $i$, let $W_i=(X_i,u_{i,1},\dots,u_{i,R})'$ and let $\hat m_{j,R}(W_i,\theta)=R^{-1}\sum_{r=1}^RM_j(X_i,u_{i,r},\theta)$.
Then, it holds under mild regularity conditions (Assumption \ref{as:cht} in Appendix A) that
\begin{align*}
n^{-1}\sum_{i=1}^n\hat m_{j,R}(W_i,\theta)\stackrel{p}{\to} E_P[m(X_i,\theta)],~
\end{align*}
uniformly in $\theta$ as $n\to\infty$ for any fixed $R$.\footnote{Moreover, under the assumption of Proposition \ref{prop:consistency}, the estimated moments are asymptotically unbiased in the sense that the empirical process $\frac{1}{\sqrt n}\sum_{i=1}^n(\hat m_{j,R}(W_i,\cdot)-E_P[m(X_i,\cdot)])$ converge weakly to a Gaussian process with zero mean.}

We state the Hausdorff consistency of level-set estimators as a proposition under a set of assumptions similar to those in \cite{ChernozhukovHanTamer2007}. For this, let $d_H(A,B)\equiv \max\{\sup_{a\in A}\inf_{b\in B}\|a-b\|,\sup_{b\in B}\inf_{a\in A}\|a-b\|\}$ denote the Hausdorff distance between two sets $A,B$.
\begin{proposition}\label{prop:consistency}
For each $c\ge 0$, let
\begin{align*}
\hat\Theta_{n,R}(c)\equiv \{\theta\in\Theta:T_{n,R}(\theta)\le c\},
\end{align*}
where $T_{n,R}(\theta)$ is defined as in \eqref{eq:test_stat} with $S(m,\Sigma)=\sum_{j=1}^J(\Sigma^{-1/2}_{jj}m_j)_+^2.$
For each $R\in\mathbb N$, let $\{c_{n,R}\}_{n=1}^\infty\subset\mathbb R_+$ be a sequence such that $c_{n,R}/n\to \infty$ and $c_{n,R}\ge \sup_{\theta\in\Theta_I}T_{n,R}(\theta)$ with probability approaching 1 (as $n\to\infty$).
Suppose that  Assumption \ref{as:cht} (in Appendix) holds.
Then, for each $R\in\mathbb N$,
\begin{align*}
d_H(\hat\Theta_{n,R}(c_{n,R}),\Theta_I)\stackrel{p}{\to}0,~~\text{as}~~n\to\infty.
\end{align*}
\end{proposition}	

While simulation based level-set estimators are consistent, the similarity to moment equalities breaks down when it comes to inference.
In models characterized by moment equalities, finite simulation draws affect a confidence region primarily through the asymptotic variance of a point estimator. However, with moment inequalities, this is no longer the case.
A noteworthy feature of the level sets based on moment inequalities is that its boundary  may depend on the sample moments in a non-standard manner. To see this, in Example \ref{ex:twomean}, write the boundary  of the level set as
\begin{align}
	\phi(\bar{m}_{n,R})=\min\big\{\bar{m}_{1,n,R}+c/\sqrt n,~\bar{m}_{2,n,R}+c/\sqrt n\big\},
\end{align}
where, for each $j$, $m_{j,n,R}=(nR)^{-1}\sum_{i=1}^n\sum_{r=1}^R1\{u_{j,i,r}<X_{j,i}\}$. Note that
$\phi(m)=\min\{m_1+c/\sqrt n,m_2+c/\sqrt n\}$  is a nonlinear function that is not differentiable at points such that $m_1=m_2$. While this function is still directionally differentiable, the (directional) derivative of $\phi$ at $E_P[m(X_i,\theta)]$ can be shown to depend on
the underlying data generating process in a discontinuous manner.
This has important consequences on inference.
Namely, even if the simulators for the moments are consistent (with a fixed simulation size), they may introduce finite sample biases
to the boundary of the level-set, which in turn may affect the performance of inference in  non-trivial ways. Below, we show numerically that this can sometimes result in severe size distortions in empirically relevant settings.

The non-standard nature of inference in moment inequality models and its finite-sample  properties have been extensively studied in the recent literature \citep{Andrews:2009aa,Andrews:2010jk,Hirano:2012qv,Chernozhukov:2013fj,Fang:2014eu}. However, to our knowledge, its consequence in relation to simulation-based inference has not been explored. One of our goals here is to quantify the effects of simulation in the context of moment inequalities and provide a practical guidance for empirical studies.

\subsection{The effects of simulated variables}\label{ssec:num_ex1}
We start with a simple numerical experiment.
Slightly generalizing Example \ref{ex:twomean},  consider $J$ moment inequality restrictions on a scalar parameter $\theta$:
\begin{equation}
	\theta-E_P[1\{u_{i,j}<X_{i,j}\}]\leq0,\quad\forall j=1,2, \cdots, J.\label{eq:mc_J_mean}
\end{equation}
The goal of the experiment is to compare the performance of two types of confidence intervals for $\theta$: one that computes the moment above analytically and the other that approximates the moment by simulation.
Let $X_i\equiv(X_{i,1},\cdots,X_{i,J})'$ be generated as an i.i.d. random vector following a $J$-dimensional standard normal distribution.
For each $i$ and $r$, let $u_{i,r}\equiv(u_{i,1,r},\cdots,u_{i,J,r})'$ be generated as a  $J$-dimensional standard normal vector independent of $X_i$. For the simulation-based confidence region, we draw, for each $i$, a random sample $\{u_{i,r}\}_{r=1}^R$ of size $R$.

 For each $c\ge 0,$ define
\begin{align}
	\mathcal C_n^{\text{Ana}}(c)&\equiv\Big(-\infty,\min_{j=1,\cdots,J}\Big\{\frac{1}{n}\sum_{i=1}^{n}\Phi(X_{j,i})+c/\sqrt{n}\Big\}\Big]\label{eq:cs_ana}\\
	\mathcal C_{n,R}^{\text{Sim}}(c)&\equiv\Big(-\infty,\min_{j=1,\cdots,J}\Big\{\frac{1}{nR}\sum_{i=1}^{n}\sum_{r=1}^{R}1\{u_{j,i,r}<X_{j,i}\}+c/\sqrt{n}\Big\}\Big],\label{eq:cs_sim}
\end{align}
where $\Phi(\cdot)$ is the cumulative distribution function of a standard normal distribution.
The first confidence region $\mathcal C_n^{\text{Ana}}$ computes the moments  analytically, while  the second confidence region $\mathcal C_{n,R}^{\text{Sim}}$ computes them using  simulation. To investigate the effect of simulation on the test statistic only, we use a common critical value $c$ for both confidence intervals.  This critical value $c$  is calculated as the $1-\alpha$ quantile of the maximum  of $J$ independent normal random variables with mean $0$ and variance $Var(\Phi(X_i))$, which corresponds to the limiting distribution of $T_n(\theta)=\sqrt n\max_{j=1,\cdots,J}\{\theta-n^{-1}\sum_{i}\Phi(X_{j,i})\}$ under the least favorable configuration (i.e. $\theta=E[1\{u_j<X_j\}]=0$ for all $j$). Note that this critical value does not account for the fact that a finite number of draws is used in \eqref{eq:cs_sim}.

Below, we report the probabilities of the confidence intervals covering the upper bound $\theta^U$ of the identified set.\footnote{
For the one-sided confidence intervals in \eqref{eq:cs_ana}-\eqref{eq:cs_sim}, covering the upper bound $\theta^U$ of the identified set is the least favorable event for covering the identified set or covering each point in the identified set.}
Table~\ref{tab:CP1}  shows the coverage probabilities of the confidence intervals for a nominal level of $1-\alpha=0.95$. We report
simulation results based on sample size $n\in\{100,250, 1000 \}$, the number of simulation draws $R\in\{1,5,10,20\}$, and the number of moment inequalities $J \in \{2,5,10,30\}$. For each setting, we generate $1000$ Monte Carlo replications.
Here, the experiments are designed to investigate the performance of the confidence intervals when relatively small numbers of draws are used. However, note that the number of simulation draws in this range
is used in practice \citep[see e.g][]{CilibertoTamer2009}.

\begin{table}[htb!]
\caption{Coverage Probabilities of $\theta^{U}=0.5$ }
\label{tab:CP1}
	\centering
\begin{tabular}{lccccc}
\hline \hline
 & Analytical & \multicolumn{4}{c}{Simulated}\tabularnewline
 \cline{3-6}
  &  & $R=1$ & $R=5$ & $R=10$ & $R=20$\tabularnewline
\cline{1-6}

  \multicolumn{6}{l}{Panel A: ($J=2$)}\\
$n=100 $ & 0.948 & 0.734 & 0.904 & 0.919 & 0.933\tabularnewline
$n=250 $ & 0.954 & 0.733 & 0.916 & 0.916 & 0.945\tabularnewline
$n=1000$  & 0.949 & 0.738 & 0.906 & 0.925 & 0.933\tabularnewline
&&&&&\\
  \multicolumn{6}{l}{Panel B: ($J=5$)}\\
$n=100 $ & 0.940 & 0.583 & 0.882 & 0.908 & 0.930\tabularnewline
$n=250 $ & 0.952 & 0.604 & 0.881 & 0.920 & 0.935\tabularnewline
$n=1000$ & 0.948 & 0.666 & 0.894 & 0.913 & 0.933\tabularnewline
&&&&&\\
  \multicolumn{6}{l}{Panel C: ($J=10$)}\\
$n=100 $ & 0.931 & 0.481 & 0.853 & 0.904 & 0.920\tabularnewline
$n=250 $ & 0.939 & 0.467 & 0.868 & 0.908 & 0.925\tabularnewline
$n=1000$ & 0.937 & 0.487 & 0.853 & 0.896 & 0.919\tabularnewline
&&&&&\\
  \multicolumn{6}{l}{Panel D: ($J=30$)}\\
$n=100 $ & 0.939 & 0.245 & 0.811 & 0.888 & 0.923\tabularnewline
$n=250 $ & 0.935 & 0.266 & 0.803 & 0.878 & 0.917\tabularnewline
$n=1000$ & 0.946 & 0.235 & 0.810 & 0.881 & 0.912\tabularnewline
\hline \hline
\end{tabular}
\end{table}

The coverage probabilities of the confidence intervals depend on the number of simulation draws $R$ and the number of inequalities $J$ in non-trivial ways.
For any $n$ and $J$, reducing the number of draws $R$ lowers the coverage probability below the nominal level, resulting in
 a size distortion. This distortion is particularly severe when the number $J$ of inequality restrictions  is large.
    For example, consider the case with $J=30$ inequalities. This setting is relevant for  empirical examples that  involve moderate to many inequalities.
	In this case,  even for $n=1000$, the simulation based confidence intervals have coverage probabilities significantly below the nominal level: 0.235 ($R$=1), 0.810 ($R$=5), 0.881 ($R$=10), and 0.912 ($R$=20) respectively. The size distortion is particularly severe when only one simulation draw is used for each $X_i$.
	 The size distortions are not as severe as this case when the number of inequalities is relatively low (e.g. $J$=2 and 5). However, the coverage probabilities are still below the nominal level in all cases.\footnote{We note that the analytical confidence interval is also undersized when $J$ is large. This is due to the fact that the critical value is also calculated by a simulation based approximation.
 However, the magnitude of the distortion is limited (at most 2\%).}

This experiment shows that simulation errors can have nontrivial impacts on the finite sample performance of the confidence intervals. In particular, size distortions can be severe in models with moderate to many moment inequalities. 
Heuristically, the size distortion arises because the boundary of the confidence interval is an irregular transformation of the underlying moment functions.
When the analytical  moment $\Phi(X_{j,i})$ is replaced with the simulation counterpart $\frac{1}{R}\sum_{r=1}^{R}1\{u_{j,i,r}<X_{j,i}\}$,
an approximation error  (of order $O_p(R^{-1/2})$) remains.  To see the effect of this, write the right end point of $\mathcal C_{n,R}^{\text{Sim}}(c)$
as
\begin{align}
	\min_{j=1,\cdots,J}\Big\{\frac{1}{n}\sum_{i=1}^{n}(\Phi(X_{j,i})+r_{j,i})+c/\sqrt{n}\Big\},\label{eq:approx_error}
\end{align}
where $r_{j,i}\equiv\frac{1}{R}\sum_{r=1}^{R}1\{u_{j,i,r}<X_{j,i}\}-\Phi(X_{j,i})$. Taking the minimum introduces a downward bias to the estimated boundary. In other words, the right end point of the confidence interval gets pushed inward, reducing the coverage probability. This bias tends to be more severe when there are many binding moment inequalities with non-negligible approximation errors.  The naive critical value does not take this into account.  In finite samples, where the variation of $\frac{1}{n}\sum_{i=1}^n r_{j,i}$ is not negligible, ignoring the effect of the simulation error may therefore result in misleading inference.

\subsubsection*{Predicted variables}
Predicted variables have similar effects on inference. Slightly modifying Example \ref{ex:twomean}, consider the restrictions
\begin{align*}
\theta-E_P[F_j(X_{j,i})]\le 0,~j=1,\dots,J,
\end{align*}
where $F_j$ is an unknown function, which can be estimated separately. As discussed earlier, it is common in empirical practice to estimate some functions (such as the profit function in Example \ref{ex:product}) before conducting inference based on the moment inequalities. Let $N_1\in \mathbb N$ denote the number of observations used to estimate $F$ in the first stage and let $\hat F_{j,N_1}$ be the first-stage estimator of $F_j$. If $F_j$ is known up to a finite-dimensional parameter $\gamma\in \mathbb R^{d_\gamma}$, it can be estimated by a parametric first-stage estimator $\hat F_{j,N_1}(\cdot)=F_j(\cdot;\hat\gamma_{N_1})$. It can also be estimated nonparametrically or with simulation. Replacing $F_j(X_{j,i})$ with its prediction $\hat F_{j,N_1}(X_{j,i})$ introduces an approximation error $r_{j,i}=\hat F_{j,N_1}(X_{j,i})-F_j(X_{j,i})$, which often satisfies  $r_{j,i}=O_p(N_1^{-\eta})$ for some $0<\eta\le 1/2$.\footnote{The rate depends on the estimator and assumptions imposed on $F$. For parametric problems, it is common to have $\eta=1/2$, while $\eta<1/2$ is common for nonparametric problems.} Therefore, ignoring the variation of the first-stage error $r_{j,i}$ can have a consequence similar to the one discussed above.

The magnitude of the first-stage error is of order $N_1^{-\eta}$, which must be evaluated in context.  For instance, the total number of observations in the first stage may be large, but if the first stage uses location fixed effects, the relevant $N_1$ is the number of observations in each location, which may be quite small in some cases.  Note that in recognition that first-stage estimation error may be an issues, \cite{Holmes2011} and \cite{HoudeNewberrySeim2017} implement a procedure similar to the first procedure we discuss (in Section~\ref{ssec:cv_correction}).

\subsubsection*{Comparison to MSM}
We note that the irregularity mentioned above does not arise in classical  moment equality models. For comparison purposes, we briefly discuss this point.
Suppose that $\theta=\theta_0$ is the unique solution to the moment equality restrictions:
\begin{align}
	E_P[m_j(X_i,\theta)]=0,~j=1,\cdots,J
\end{align}
for $J\ge d_\theta$. A  method of simulated moments (MSM) estimator $\hat\theta_{n,R}$ is defined as
\begin{align}
	\hat\theta_{n,R}=\text{argmin}_{\theta\in\Theta}\hat E_n[\hat m_R(X_i,\theta)]'\hat \Sigma_{n,R}(\theta)^{-1}\hat E_n[\hat m_R(X_i,\theta)],
\end{align}
where $\hat m_R(X_i,\theta)=(\hat m_{1,R}(X_i,\theta),\cdots,\hat m_{J,R}(X_i,\theta))'$.
Confidence regions  can be constructed around $\hat\theta_n$.
Under regularity conditions that ensure asymptotic normality, the MSM estimator depends on the sample moments in a regular manner.

Let $D_0=\nabla_\theta E_P[m(X_i,\theta_0)]$ and $\Omega_0=E_P[m(X_i,\theta_0)m(X_i,\theta_0)']$. It is well-known that, under regularity conditions,  the MSM estimator is asymptotically linear in the sense that
\begin{align}
	\hat\theta_{n,R}=\theta_0+\frac{1}{nR}\sum_{i=1}^n\sum_{r=1}^R&l(X_i,u_{i,r},\theta_0)+o_p(n^{-1/2}),\notag\\
	&l(x,u,\theta_0)=(D_0'\Omega_0D_0)^{-1}D_0'\Omega_0M(x,u,\theta_0),
\end{align}
where the \emph{influence function} $l(x,u,\theta_0)$ has zero mean and measures the (first-order) effect of each observation $(x,u)=(X_i,u_{i,r})$ on the variation of the estimator and hence determines the asymptotic variance of $\hat\theta_{n,R}$ \citep[see e.g][]{Newey:1994aa,GourierouxMontfort1996}. In sum, the first-order effect of simulation enters only the asymptotic variance of the estimator.

To see the effect of simulation on confidence intervals, consider a slight modification of Example \ref{ex:twomean} where $\theta_0$ solves
\begin{align}
	\theta_0-E_P[1\{u_j< X_j\}]= 0,~j=1,2.
\end{align}
It is straightforward to show that an MSM estimator with equal weights on the moment conditions is given by $\hat\theta_{n,R}=(\bar{m}_{1,n,R}+\bar{m}_{2,n,R})/2$. A one-sided confidence interval on $\theta_0$ can be constructed as
\begin{align}
\mathcal C^{\text{MSM}}_{n,R}(c)\equiv(-\infty,\hat\theta_{n,R}+c/\sqrt n]=\Big(-\infty,\frac{\bar{m}_{1,n,R}+\bar{m}_{2,n,R}}{2}+c/\sqrt n\Big],
\end{align}
where $c$ is the $1-\alpha$ quantile of the asymptotic (normal) distribution of the MSM estimator. The form of the confidence interval
shows that (i) the boundary of the confidence interval depends on the sample moments in a smooth manner; and hence (ii) simulation errors in the MSM estimator affects $\mathcal C^{\text{MSM}}_{n,R}$, but it can easily be accounted for by adjusting $c.$ That is, the increased variance of the MSM estimator can be accommodated using a suitable estimator of the asymptotic variance that accounts for $R$ being finite.
For moment inequalities, however, it turns out that this type of variance correction is not enough. As we show in Section \ref{sec:size_correction}, one also needs to account for a potential bias in the estimated boundary.

\subsubsection*{The effects of simulation and a common empirical feature}\label{ssec:num_ex2}
Before proceeding further, we illustrate a common feature of empirical examples that can potentially cause serious size distortions if simulation is used naively. The distortion can be particularly severe when the number of draws is small. To highlight this,
 we design a data generating process based on the entry game in Example \ref{ex:entry}.

Let $Z_{i}=(Z_{1,i},Z_{2,i})$ collect the observable characteristics of the two firms. We let $Z_i$ be generated as a discrete random vector supported on a finite set $\mathcal Z=\{z_k,k=1,\cdots,K\}$. Table~\ref{tab:PZ} gives the distribution and support of $Z_i$. In empirical studies, it is a common practice to classify continuous state variables into a finite number of bins. When such discretization is used, some bins may contain a limited number of observations. The distribution in Table \ref{tab:PZ} emulates this feature by assigning low probabilities to some bins.

\begin{table}[htbp]
\caption{Probability Distribution of $Z$}
\label{tab:PZ}
	\centering
\begin{tabular}{rrrrrr}
\hline
$z_{1}$ & $-0.1$ & $-0.5$ & $0$ & $0.5$ & $1$\tabularnewline
\cline{2-6}
$P(Z_{1}=z_1)$ & $0.1$ & $0.1$ & $0.1$ & $0.1$ & $0.6$\tabularnewline
\hline
 &&&&&\\
\cline{1-4}
$z_{2}$ & $-0.5$ & $0$ & $0.5$ &  & \tabularnewline
\cline{2-4}
$P(Z_{2}=z_2)$ & $0.1$ & $0.8$ & $0.1$ &  & \tabularnewline
\cline{1-4}
\end{tabular}

\vspace{0.1in}
Note: $Z_1$ and $Z_2$ are independent.
\end{table}

We generate unobservable characteristics $u_i=(u_{1,i},u_{2,i})$  as a bivariate standard normal vector independent of $Z_{i}$. For simplicity, we assume symmetry between the firms $(\beta=\beta_j,\Delta=\Delta_j,j=1,2)$ and set $\theta\equiv(\beta,\Delta)=(0.9,-0.5)$.
For some values of $u_i$, the model predicts both $Y_i=(1,0)$ and $(0,1)$ as multiple equilibria.  If this is the case, we
select the equilibrium $Y_i=(1,0)$ with probability $0.7$ (independent of $(Z_i,u_i)$).
The knowledge on the
selection mechanism is not used for inference, and hence the agnosticism
about the equilibrium selection rule leads to partial identification
of parameters. Figure~\ref{fig:Identifiedset} shows the identified set $\Theta_I$ for $\theta$ based only on the moment inequalities \eqref{eq:entry3}-\eqref{eq:entry4}.\footnote{To focus on the non-standard effects through the moment inequalities, we drop the moment equality restrictions in this exercise. Due to the presence of covariates with 15 support points,  there are a total of $30$ unconditional moment inequalities.} In what follows, we report coverage probabilities on one of the extreme points $\theta^{U}=(0.8880,-0.4015)$ of $\Theta_I$, which gives the
upper bound on the competitive effect parameter $\Delta$.

\begin{figure}[htbp]
\centering
\caption{Identified Set: Entry Game}
\label{fig:Identifiedset}
\includegraphics[scale=0.6]{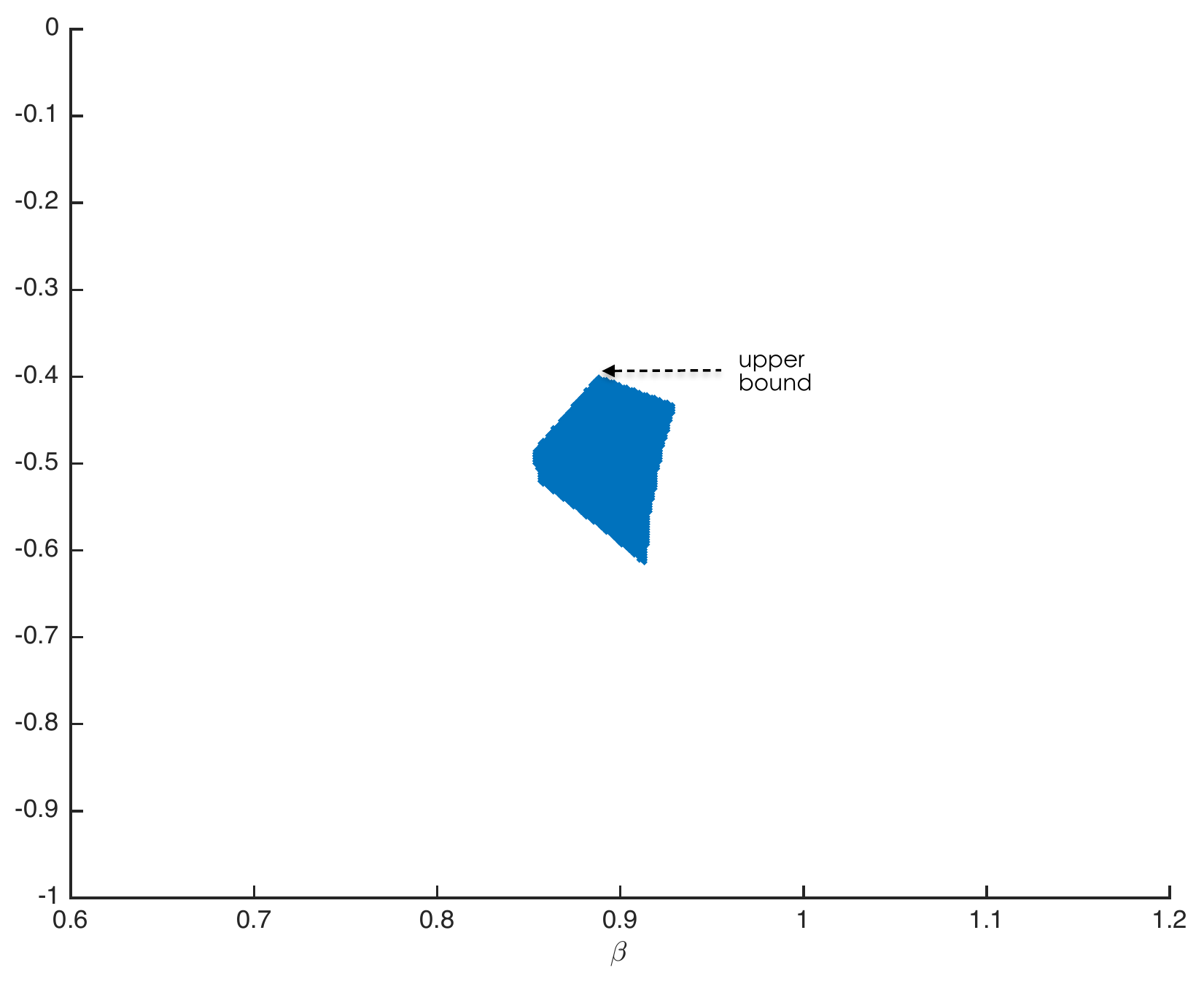}
\end{figure}

Using the specification above and instrument functions  $1\{Z_i=z_k\},k=1,\cdots,K$, we  transform the conditional moment inequalities in \eqref{eq:entry3}-\eqref{eq:entry4} into the following
unconditional moment inequalities:
\begin{align}
E\big[\big(1\{Y_i=(0,1)\} - H_1(Z_{i};\theta)\big)1\{Z_i=z_k\}\big]&\le 0\\
E\big[\big(H_2(Z_{i};\theta)  -1\{Y_i=(0,1)\}\big)1\{Z_i=z_k\}\big]&\le 0,\notag
\end{align}
where the entry probabilities
\begin{align}
	H_1(Z_{i};\theta)&\equiv P(u_1\le -Z_{1,i}'\beta_1-\Delta_1,~u_2> -Z_{2,i}'\beta_2|Z_i)\\
H_2(Z_{i};\theta)&\equiv	P(u_1\le -Z_{1,i}'\beta_1-\Delta_1,u_2>-Z_{2,i}'\beta_2-\Delta_2|Z_i)\notag\\
		&\qquad\qquad+P(u_1\le -Z_{1,i}'\beta_1,-Z_{2,i}'\beta_2\le u_2\le -Z_{2,i}'\beta_2-\Delta_2|Z_i)
\end{align}
are calculated using either an analytical expression or a frequency simulator.
For example, using the parametric specification, $H_{1}(Z_i,\theta)$ may be computed analytically as  $\Phi(-Z_{1,i}'\beta)\times \Phi(Z_{2,i}'\beta)$.  Alternatively, using a simulator one may compute the same object as $ R^{-1}\sum_{r=1}^R1\{u_{1,i,r}\le -Z_{1,i}\beta,u_{2,i,r}>-Z_{2,i}\beta\}$, where $(u_{1,i,r},u_{2,i,r}),r=1,\cdots,R$ are drawn from the bivariate standard normal distribution.

Our benchmark inference procedure is implemented as follows. The confidence
region takes the form: $CS_{n}=\{\theta\in\Theta: T_{n,R}(\theta)\le c_{n}(\theta)\}$,
where $T_{n,R}$ is the statistic proposed by \cite{Rosen:2008aa}
and further refined by \cite{AndrewsJia2012}:
\begin{eqnarray*}
T_{n,R}(\theta)=\inf_{t\in \mathbb R_{+}^{J}}(\sqrt n\bar{m}_{n,R}(\theta)-t)'\tilde{\Sigma}_{n,R}^{-1}(\theta)(\sqrt n\bar{m}_{n,R}(\theta)-t),
\end{eqnarray*}
where $\tilde \Sigma_{n,R}$ is a suitable estimator of the asymptotic variance of the moments. The  critical value $c_{n}(\theta)$ is computed using a bootstrap procedure combined with the generalized moment selection (GMS) procedure  \citep{Andrews:2010jk,AndrewsJia2012}.\footnote{We use $\tilde{\Sigma}_{n}(\theta)=\hat{\Sigma}_{n}(\theta)+\max\{0.012-\text{det}(\hat{\Omega}_{n}(\theta)),0\}\hat{D}_{n}(\theta)$, where $\hat{D}_{n}(\theta)=Diag(\hat{\Sigma}_{n}(\theta))$ and $\hat{\Omega}_{n}(\theta)=\hat{D}_{n}^{-1/2}(\theta)\hat{\Sigma}_{n}(\theta)\hat{D}_{n}^{-1/2}(\theta)$.
We set $\kappa_{n}=n^{1/16}$ for the GMS parameter.
}
For  details on the GMS procedure, we refer to the references above, but we briefly describe its mechanism to highlight the  potential effects of simulation on this procedure. The key idea of the GMS is to compute the
critical value by selecting the moments that are relevant to the asymptotic
null distribution of the test. For example, in an implementation of the GMS, the $j$-th moment inequality is ``selected''
and used to calculate $c_n$ if the studentized moment is smaller than a tuning parameter $\kappa_{n}$, i.e. $\frac{\sqrt n\bar{m}_{j,n,R}(\theta)}{\hat{\sigma}_{j,n,R}(\theta)}\le \kappa_{n}$,
where $\hat{\sigma}_{j,n,R}^{2}(\theta)$ is the $j$-th diagonal element of
$\hat{\Sigma}_{n,R}(\theta)$. The critical
value is then computed as the $1-\alpha$ quantile of the bootstrapped
statistic where the sample moments are replaced with the bootstrap analog of the selected moments. We note here that the simulated moments could also potentially affect the GMS step, but its consequence is not immediately clear.

\begin{table}[htbp]
\centering
\caption{Coverage Probabilities for $\theta^U$}
\label{tab:Coverage}
\begin{tabular}{lccccc}
\hline\hline
   & Analytical & \multicolumn{4}{c}{Simulated}\tabularnewline
\cline{3-6}
   &  & $R=1$ & $R=5$ & $R=10$ & $R=20$\tabularnewline
\hline
Panel A: ($n=250$) & & & & & \\
  Coverage prob & $0.921$ & $0.391$ & $0.453$ & $0.449$ & $0.454$\tabularnewline

  $\#$ of times differ in selection &  & $997$ & $858$ & $755$ & $666$\tabularnewline
&&&&&\\
Panel B: ($n=500$) & & & & & \\
  Coverage prob & $0.945$ & $0.814$ & $0.903$ & $0.906$ & $0.909$\tabularnewline

 $\#$ of times differ in selection &  & $979$ & $766$ & $587$ & $448$\tabularnewline

&&&&&\\
Panel C: ($n=1000$) & & & & & \\
 Coverage prob & $0.961$ & $0.833$ & $0.947$ & $0.952$ & $0.960$\tabularnewline

 $\#$ of times differ in selection &  & $982$ & $800$ & $638$ & $475$\tabularnewline

&&&&&\\
Panel D: ($n=2000$) & & & & & \\
  Coverage prob & $0.953$ & $0.756$ & $0.935$ & $0.941$ & $0.949$\tabularnewline

 $\#$ of times differ in selection &  & $990$ & $831$ & $709$ & $539$\tabularnewline
\hline \hline
\end{tabular}
\end{table}

Table~\ref{tab:Coverage} reports simulation results based on sample size $n\in\{250,500, 1000,2000\}$
and the number of simulation draws $R=\{1,5,10,20\}$. We simulate $1000$
datasets for each setting. The coverage probabilities of the confidence regions with nominal level  $95\%$ are evaluated for
the true upper bound  $\theta^{U}$.

The main finding is that, for small sample size and simulation size, the coverage probabilities  of the confidence region are distorted  in a significant way. In particular,
with $n=250$, the coverage probabilities of the simulation-based confidence regions vary from 39.1\% to 45.3\%, significantly below the nominal level.
Even for a moderate sample size, these confidence regions
exhibit size distortions when the number of draws $R$ is small. For example,
for $n=2000$, the coverage probabilities  of the simulated-based confidence region is
75.6\% when only a single draw is used. However, this gets improved (to 93.5\%) with $R=5$ draws. Table \ref{tab:Coverage} suggests that, with $R=20$, the coverage probabilities of simulation based confidence regions are close to those of of the analytical ones in relatively large samples ($n=1000,2000$).

This experiment suggests that the test statistic may not be able to benefit from averaging of the simulation errors when the number of observations in each bin  is small.
Combined with the irregular nature of the test statistic,  this can result in a severe size distortion. The presence of bins with small numbers of observations is common in empirical settings. In such cases, care must be taken.\footnote{We conjecture that the same comment applies to conditional moment inequality models if the  essential sample size is small (i.e. small $n$ and bandwidth $h^{d_Z}$).}
One possibility may be to increase the number of simulation draws for such bins. However, the choice of the number of draws becomes an arbitrary component of inference. Another possibility is to modify the procedure to explicitly account for the effects of the simulation. In the next section, we explore several possibilities along this line.

As a final remark, we note that
the simulation not only affects the
value of the test statistic but also the general moment selection
procedure.
Table \ref{tab:Coverage} reports the number of times (out of 1000) the set of the selected moments differ between the simulated and analytical methods.  This difference arises because the simulation  adds
perturbation to the studentized moment. With a limited number of draws, this effect on the GMS may not be negligible.  As a result, some of the ``selected''
inequalities (i.e. $\sqrt{n}\bar{m}_{j,n}(\theta)/\hat{\sigma}_{j,n}(\theta)\ge -\kappa_{n}$) according to the GMS based on the analytical moment may be discarded (i.e. $\sqrt{n}\bar{m}_{j,n,R}(\theta)/\hat{\sigma}_{j,n,R}(\theta)<- \kappa_{n}$) under the GMS based on the simulated moment, and vice versa.

\section{Inference methods with corrections for simulation errors}
\label{sec:size_correction}
We consider two methods to account for the simulation errors. One is to correct the critical value. The other is to correct or ``regularize'' the test statistic. The former method closely follows the recent development on the moment inequality literature, and hence we keep its discussion minimal. The second method is a novel approach and has several attractive features for simulation based methods. It combines a simple bias correction method with a bootstrap critical value based on a regular (and differentiable) functional of a random vector that is asymptotically normal. To our knowledge, the uniform validity of such a method in the context of simulation-based inference is new to the literature.

\subsection{Critical value correction methods}\label{ssec:cv_correction}
Commonly used inference methods provide approximations to the distribution of the test statistic in \eqref{eq:test_stat}. Many of these methods are based on resampling techniques.
As we saw in the previous section, ignoring the variation due to simulation can result in poor inference. Recall that, in the numerical experiment in Section \ref{ssec:num_ex1}, the size distortion occurred because
the naive critical value did not take into account the bias that was due to the extra variation from the approximation error $\frac{1}{n}\sum_{i=1}^n r_{j,i}$ (in \eqref{eq:approx_error}).
A straightforward way to correct this effect is to adjust the critical value.
One can achieve this by a bootstrap critical value that resamples the simulated variables $u_{i,r}$ along with the original sample $X_i$ across bootstrap replications.
Specifically, let $\{X^*_i,i=1,\dots,n\}$ be a bootstrap sample drawn with replacement from the empirical distribution. For each $i$,  let $\{u^*_{i,r},r=1,\dots,R\}$ be drawn from $P(\cdot|X^*_i)$.

We then let $\hat m^*_{R}(W^*_i,\theta)=R^{-1}\sum_{r=1}^RM(X^*_i,u^*_{i,r},\theta)$ be the simulation approximation to the (conditional) moment function in the $i$-th bootstrap sample.
Define
\begin{align}
	\bar m^*_{n,R}(\theta)\equiv \frac{1}{n}\sum_{i=1}^n \hat m^*_{R}(W^*_i,\theta).
\end{align}
This resampled moment is constructed so that it mimics the behavior of $\bar m_{n,R}$ including the variation due to the simulated variables.
 The final step is to use the resampled moment above when one applies an existing method for computing the critical value.\footnote{Due to the non-standard nature of $S$, consistently approximating the limiting distribution of $T_{n,R}$ uniformly over a large class of DGPs is not possible in general. Hence, the existing methods resort to conservative distortions.} For example, consider the generalized moment selection  (GMS) in \cite{Andrews:2010jk}.
 Let $\hat D_{n,R}(\theta)=\text{diag}(\hat\Sigma_{n,R}(\theta))$ and let $\hat \Omega_{n,R}(\theta)$ be an estimator of the correlation matrix, i.e. $\hat \Omega_{n,R}(\theta)=\hat D_{n,R}(\theta)^{-1/2}\hat\Sigma_{n,R}(\theta)\hat D_{n,R}(\theta)^{-1/2}.$  This procedure uses the $1-\alpha$  quantile $c_{\kappa,n,R}(\theta)$ of the following statistic
\begin{align}
	T^*_{n,R}(\theta)=S\big(\hat D_{n,R}(\theta)^{-1/2}Z^*_{n,R}(\theta)+\varphi(\xi_{n,R}(\theta)),\hat \Omega_{n,R}(\theta)\big),\label{eq:test_stat}
\end{align}
where $Z^*_{n,R}(\theta)=\sqrt n(\bar m^*_{n,R}(\theta)-\bar m_{n,R}(\theta))$ is the bootstrapped empirical process,
and  $\xi_{n,R}$ is a $J\times 1$ vector whose components are
\begin{align}
	\xi_{j,n,R}(\theta)\equiv \kappa_n^{-1}\sqrt n\frac{\bar m_{j,n,R}(\theta)}{\hat\sigma_{j,n,R}(\theta)},~j=1,\dots,J.\label{eq:xis}
\end{align}
The generalized moment selection (GMS) function  $\varphi$ then selects the inequalities
that are relevant for the inference for $\theta$ based on $\xi_{n,R}$ (see \cite{Andrews:2010jk} for details).
The resulting confidence region is
\begin{align}
	\mathcal C_{n,R}^{\text{CV}}=\{\theta\in\Theta:T_{n,R}(\theta)\le c_{\kappa,n,R}(\theta)\}.\label{eq:conf_region_AS}
\end{align}

In Example \ref{ex:twomean}, applying this method with a $t$-test based GMS function ($\varphi(\xi_j)=0$ if $\xi_j\ge -1$ and $=-\infty$ otherwise) yields the following confidence interval
\begin{multline}
	\mathcal C_{n,R}^{\text{CV}}=\Big\{\theta\in\mathbb R:\sqrt n\max\Big\{\frac{\bar m_{1,n,R}(\theta)}{\hat\sigma_{1,n,R}},\frac{\bar m_{2,n,R}(\theta)}{\hat\sigma_{2,n,R}}\Big\}\le c_{\kappa,n,R}\Big\}\\
	=\Big(-\infty,\min_{j=1,2}\Big\{\frac{1}{nR}\sum_{i=1}^n\sum_{r=1}^R1\{u_{j,i,r}<X_{j,i}\}+c_{\kappa,n,R}\frac{\hat\sigma_{j,n,R}}{\sqrt n} \Big\}\Big],\label{eq:twomean_as}
\end{multline}
where the critical value $c_{\kappa,n,R}$ is the $1-\alpha$ quantile of the  bootstrapped statistic\footnote{The critical value and the standard deviations in this example are not indexed by $\theta$ because the re-centered moment functions $\bar m_{j,n,R}(\theta)-E[m_j(X_i,\theta)],j=1,\dots,J$ in the test statistic do not depend on $\theta$.}
\begin{align}
	T^*_{n,R}(\theta)&\equiv \max_{j\in\{j:\xi_{j,n,R}\ge -1\}}\frac{\sqrt n(\bar m^*_{j,n,R}(\theta)-\bar m_{j,n,R}(\theta))}{\hat\sigma_{j,n,R}}\notag\\
	&= \max_{j\in\{j:\xi_{j,n,R}\ge -1\}}\frac{\frac{1}{\sqrt nR}\sum_{i=1}^n\sum_{r=1}^R1\{u^*_{j,i,r}<X^*_{j,i}\}-1\{u_{j,i,r}<X_{j,i}\}}{\hat\sigma_{j,n,R}}.\label{eq:bootstat}
\end{align}
The confidence interval in this example is the intersection of the  sample (upper) bounds that are suitably expanded. The amount of the expansion  $c_{\kappa,n,R}\frac{\hat\sigma_{j,n,R}}{\sqrt n}$ accounts for the variation of each simulated moment. This type of confidence interval is also considered in the context of conditional moment restrictions  in \cite{Chernozhukov:2013fj}.

\begin{remark}\rm
The method described above accounts for the simulation variation contained in the sample moment $\bar m_{n,R}$. Hence, one can expect that it mitigates the problem we saw in Section \ref{ssec:num_ex2}. However, it does not account for the second channel through which simulation can affect the performance of the confidence region. This is through the rescaled moments $\xi_{j,n,R},j=1,\cdots, J$ used in the GMS function. As we saw in Table \ref{tab:Coverage}, the effect of simulation on the selected set of moments is nontrivial.
However, since the critical value may depend on $\xi_{n,R}$ in a discontinuos manner (see the $t$-test based on GMS above), the  analysis of the effect of the simulation error on the  coverage probability through this channel is complex.
We therefore do not pursue the correction of the effect of simulation through $\xi_{n,R}$. In contrast, our alternative inference method in the next section admits a straightforward way to correct the effect of simulation.
\end{remark}

\subsection{Regularization of test statistics}
In addition to the modification of the existing  methods, we propose a novel and computationally simple inference procedure that accounts for simulation. The key
 idea is as follows. Recall that
commonly used test statistics take the form: $T_{n,R}(\theta)=S(\sqrt n\bar m_{n,R}(\theta),\hat\Sigma_{n,R}(\theta))$ whose irregular behavior arises due to the non-smoothness of $S$.
Our procedure replaces $S$
by a smooth approximation $S_\mu$, which satisfies certain regularity conditions. This approach has several attractive features. First, the proposed method has a uniform  approximation property. That is, for any $(m,\Sigma)$, $|S_\mu(m,\Sigma)- S(m,\Sigma)|\le \beta\mu$  for a known uniform constant $\beta$ and the degree of smoothness $\mu>0$, which is chosen by the researcher.  Accounting for the approximation error  is then straightforward because $\beta$ is known. Second, the approximated test statistic $\tilde T_{n,R}(\theta)\equiv S_\mu(\sqrt n\bar m_{n,R}(\theta),\hat\Sigma_{n,R}(\theta))$ obeys standard limit theorems uniformly over a large class of DGPs and over a range of values for $\mu$. This in turn allows one to employ a standard resampling method such as bootstrap to calculate the critical value. Finally, smooth approximations to a wide class of non-smooth convex functions are available thanks to the recent developments in the non-smooth convex optimization literature \citep[see e.g.][]{Nesterov:2005qv,Beck:2012xy}.
Using these results, we provide functional forms of smooth approximations to  some of the commonly used test statistics. The idea of  regularizing  test statistics (or estimated bounds) also appears in related contexts \citep[][]{HaileTamer2003,Chernozhukov:2015aa,Kaido:2016aa,Masten:2017aa}. Our contribution here is to show its uniform validity in the context of inference with simulated variables.
 A potential price for this computationally simple method is the possibility of inference becoming conservative for some choice of the smoothing parameter. We will examine this point numerically in Section \ref{sec:montecarlo}.

Below, we illustrate our approach using Example \ref{ex:twomean}.
\setcounter{example}{0}
\begin{example}[Intersection bounds (continued)]\rm
	Recall the confidence interval in \eqref{eq:twomean3}:
\begin{align}
	\mathcal C_{n,R}^{\text{Sim}}
	&=\{\theta:S\big(\sqrt n (\bar m_{1,n,R}(\theta),\bar m_{2,n,R}(\theta))',\hat\Sigma_{n,R}(\theta)\big)\le c_{n,R}\},\label{eq:two_mean4}
\end{align}	
where $S((m_1,m_2)',\Sigma)=\max\{m_1,m_2\}$.
Consider replacing $S$ with a smooth approximation. Specifically, for $\mu>0$, define
\begin{equation}
	\tilde{\mathcal C}_{n,R}
=\{\theta:S_\mu\big(\sqrt n (\bar m_{1,n,R}(\theta), \bar m_{2,n,R}(\theta)),\hat\Sigma_{n,R}(\theta)\big)\le \tilde c_{\mu,n,R}\},\label{eq:two_mean5}
\end{equation}
where $S_\mu((m_1,m_2)',\Sigma)\equiv \mu\ln(\exp(\frac{m_1}{\mu})+\exp(\frac{m_2}{\mu}))$ replaces the maximum function.
We then calculate our critical value $\tilde c_{\mu,n,R}$  by
\begin{align}
	\tilde c_{\mu,n,R}=c_{1-\alpha,R}+\sqrt n\mu\ln 2.\label{eq:two_mean6}
\end{align}
This critical value consists of two terms.
The first term,
 $c_{1-\alpha,R}$, is an approximation to the $1-\alpha$ quantile of the root (centered and rescaled statistic):\footnote{In this example, $Z_n$ does depend on $\theta$. Hence, its $1-\alpha$ quantile does not depend on $\theta$ either.}
 \begin{align}
	Z_n=\sqrt n\big(S_\mu((\bar m_{1,n,R}(\theta),\bar m_{2,n,R}(\theta))',\hat \Sigma_{n,R}(\theta))-S_\mu((E_P[m_1(X_i,\theta)],E_P[ m_{2}(X_i,\theta)])',\Sigma_P(\theta))\big).\label{eq:zndef}
 \end{align}
 The second term in \eqref{eq:two_mean6} is a bias-correction term that accounts for the approximation error (or population-level bias) that arises when we replace $S(E_P[m(X_i,\theta)],\Sigma_P(\theta))$ by its smooth counterpart $S_\mu(E_P[m(X_i,\theta)],\Sigma_P(\theta))$.
The quantile $c_{1-\alpha,R}$ can be approximated by a bootstrap procedure that resamples an analog of $Z_n$ in \eqref{eq:zndef} (see Algorithm 1 below).

The confidence interval in \eqref{eq:two_mean5} is asymptotically valid over a wide class of data generating processes and choices of $\mu$.
We sketch the argument  below and defer the formal proof to the sequel.

Let $\theta=\min\{E_P[\Phi(X_{1,i})],E_P[\Phi(X_{2,i})]\}$ be the upper boundary point of the identified set.
Then,  $S(E_P[m(X_i,\theta)],\Sigma_P(\theta))=0.$
Let $0<\underline{M}<\overline M<\infty.$
Then, uniformly in $\mu\in [\underline{M},\overline M]$ and in $P$ over a class of distributions specified below,
the (least favorable) coverage probability is
\begin{align}
	P&(\theta \in \tilde{\mathcal C}_{n,R})\notag\\
	&=P\big(S_\mu\big(\sqrt n (\bar m_{1,n,R}(\theta), \bar m_{2,n,R}(\theta)),\hat\Sigma_{n,R}(\theta)\big)\le \tilde c_{\mu,n,R})\\
	&=P\big(Z_n-\sqrt n(S_\mu(E_P[m(X_i,\theta)],\Sigma_P(\theta))-S(E_P[m(X_i,\theta)],\Sigma_P(\theta)))\le \tilde c_{\mu,n,R}\big)\notag\\
	&\ge P(Z_n+  \sqrt n\mu\ln 2\le \tilde c_{\mu,n,R})\notag\\
	&=P(Z_n\le c_{1-\alpha,R})\to 1-\alpha,
\end{align}
where we used  \eqref{eq:zndef}, $|S_\mu(m)-S(m)|\le \beta\mu$ with $\beta=\ln 2$ (see Table \ref{tab:musmooth}). The convergence in the last step follows from the  argument below.

Note that $S_\mu$ is differentiable with the following derivative:
\begin{align}
	DS_\mu[m](h)= \sum_{j=1}^2w_{j}(m,\mu)h_j,~w_j(m,\mu)=\frac{e^{m_j/\mu}}{\sum_{j=1}^2e^{m_j/\mu}}.
\end{align}
It can be shown that $DS_\mu$ is Lipschitz continuous in $m$ with a Lipschitz constant $\mu^{-1}.$
Hence, by the mean-value theorem,
\begin{align}
	Z_n&= DS_\mu[\bar m_{n,R}^*](\sqrt n(\bar m_{n,R}-E_P[m(X_i)]))\\
	&=DS_\mu[E_P[m(X_i)]]\sqrt n(\bar m_{n,R}-E_P[m(X_i)])+\frac{1}{\mu}r_n,
\end{align}
where $r_n=o_P(1)$ uniformly in $P$ as $n\to\infty$ with $R$ fixed.
Hence, $Z_n$ converges in distribution to some limit $Z$ uniformly in $\mu$ over a compact set (not containing 0).

Now, let us compare this to a method without any regularization of $S$.
 The coverage probability of the confidence interval in \eqref{eq:two_mean4} is
\begin{align}
	P(\theta \in \mathcal C_{n,R}^{\text{Sim}})&=P\big(\sqrt n(S(\bar m_{n,R}(\theta),\hat \Sigma_{n,R}(\theta))-S(E_P[m(X_i,\theta)],\Sigma_P(\theta)))\le  c_{n,R} \big),
\end{align}
 The transformation $S$ is not differentiable but can be shown to be directionally differentiable with the following directional derivative:
\begin{align}
	DS[m](h)=\min_{j\in \mathcal J^*(m)}h_j,~\mathcal J^*(m)=\{j:m_j=\min\{m_1,m_2\}\}.
\end{align}
Observe that the directional derivative is non-linear in $h$ (but only positively homogeneous). More importantly, the directional derivative
depends on $m$ and hence the underlying data generating process in a discontinuous manner.  This is because the set of ``active'' inequalities, $\mathcal J^*(m)$, depends on $m$ discontinuously.
This in turn implies that the limiting distribution of the test statistic is discontinuous in the underlying DGP.\footnote{This follows from a  $\delta$-method for directionally differentiable functions \citep{Shapiro:1991aa}. The discontinuity in the limiting distribution of the statistic can  be shown without using the directional derivative \citep{Andrews:2010jk}. We use the directional derivative to make a comparison to the method based on the $\mu$-smooth approximation and standard $\delta$-method.} Hence, a small perturbation of the underlying data generating process may result in a significant change in the distribution of the statistic.
When simulation  is used to replace the population moments, this therefore could affect the behavior of the statistic in non-trivial ways.
This feature  motivated the vast literature on moment inequalities that corrects the critical value whenever this type of discontinuity is a concern \citep[see e.g][]{Andrews:2010jk,Chernozhukov:2013fj,Fang:2014eu,Romano:2014aa}.

In contrast, our approach first replaces the $S(\bar m_{n,R}(\theta),\hat\Sigma_{n,R}(\theta))$ by a smooth approximation $S_\mu(\bar m_{n,R}(\theta),\hat\Sigma_{n,R}(\theta))$, while correcting for the population level bias.
Since the limiting distribution of $S_\mu(\bar m_{n,R}(\theta),\hat\Sigma_{n,R}(\theta))$ depends on the underlying distribution in a smooth manner, there is no need to  correct the critical value $c_{1-\alpha,R}$.
In short, our approach regularizes the behavior of the statistic, while the vast literature regularizes the critical value to ensure the uniform validity of inference.

The motivations for this approach are two-fold. First, it is straightforward to show that the proposed method is
uniformly valid under the large $n$ asymptotics with a fixed simulation size $R$. The proposed method  combines the standard $\delta$-method with
 a bias correction (at the population level). The method and its uniform validity may be of independent interest outside the context of simulation based inference.
Second, the inference method is simple and can be implemented by a standard bootstrap procedure.
 This is attractive as accounting for
the effects of simulation errors on moment selection procedures or Bonferroni-correction methods in the existing literature may be non-trivial.
\end{example}

In general, we define our confidence set by
\begin{align}
	\tilde{\mathcal C}_{n,R}\equiv\{\theta:\tilde T_{n,R}(\theta)\le \tilde c_{n,R}(\theta)\},\label{eq:def_cs}
\end{align}
where the test statistic and the critical value are calculated as follows
\begin{align}
\tilde T_{n,R}(\theta)&\equiv S_\mu(\sqrt n\bar m_{n,R}(\theta),\hat\Sigma_{n,R}(\theta)),\label{eq:test_stat2}\\
\tilde c_{n,R}(\theta)&\equiv c_{n,R,1-\alpha}(\theta)+\sqrt n\mu\beta,\label{eq:crit_val}
\end{align}
where $c_{n,R,1-\alpha}$ is an estimate of the $1-\alpha$ quantile of
\begin{align}
Z_{\mu,n,R}=S_\mu(\sqrt n\bar m_{n,R}(\theta),\hat\Sigma_{n,R}(\theta))- S_\mu(\sqrt nE_P[m(X_i,\theta)],\Sigma(\theta)).\label{eq:def_zmunr}
\end{align}
Here, $S_\mu$ is an approximation to $S$, which has  smoothness properties that are useful for analyzing and correcting the behavior of the test statistic
in the presence of simulated variables. We introduce the following notion of approximation based on \cite{Beck:2012xy}.\footnote{The third condition in Definition \ref{def:musmooth} is not
required in \cite{Beck:2012xy} but is satisfied by all $\mu$-smooth approximations we use in this paper and is useful for establishing asymptotic results.}

\begin{definition}[$\mu$-smooth approximation]\label{def:musmooth}
Let $\phi:\mathbb R^J\to (-\infty,\infty]$ be a closed and proper convex function and let $M\subseteq \text{dom}(\phi)$	be a closed convex set.
A function $\phi_\mu:\mathbb R^J\to (-\infty,\infty)$ is said to be a \emph{$\mu$-smooth approximation} of $\phi$ with parameters $(\alpha,\beta,K)$ if the following conditions hold:

\noindent
(i) $\phi(m)-\beta_1 \mu\le \phi_\mu(m)\le \phi(m)+\beta_2\mu$ for some $\beta_1,\beta_2$ satisfying $\beta_1+\beta_2=\beta>0$;

\noindent
(ii) $\phi_\mu$ has a derivative $D \phi_\mu[m](\cdot)$ such that
\begin{align}
	\big\|D \phi_\mu[m]-D\phi_\mu[m']\big\|^*\le (K+\frac{\alpha}{\mu})\|m-m'\|,~\forall m,m'\in M,
\end{align}
for some $K\ge 0$ and $\alpha>0$;

\noindent
(iii) For each $m$, $(\mu,h)\mapsto D\phi_\mu[m](h)$ is continuous.
\end{definition}
Definition \ref{def:musmooth} (i) requires that $\phi_\mu$ has a uniform approximation property, which is the key condition for bias correction.
Definition \ref{def:musmooth} (ii) requires that $\phi_\mu$'s derivative is Lipschitz continuous, which plays a role in making the limiting distribution of $Z_{\mu,n,R}$   depend continuously on the underlying DGP.
Given this definition, we require $S$ and $S_\mu$ to satisfy the following condition. For this, let $\mathbb P^J$ be the set of $J\times J$ symmetric positive semi-definite matrices.
\begin{assumption}\label{as:onSmu}
	(i) The index function $S:\mathbb R^J\times \mathbb P^J\to\mathbb R_+$ satisfies Assumptions $1-6$ in \cite{Andrews:2010jk}. For any $a>0$ and $(m,\Sigma)\in \mathbb R^J\times \mathbb P^J$, $S(am,\Sigma)=a^{\chi}\phi(V^{-1/2}m)$ for a proper convex function $\phi:\mathbb R^J\to\mathbb R_+$ with $\chi=1$, where $V=\text{diag}(\Sigma)$; (ii) For each $\mu\in[\underline M,\overline M]$, the map $S_\mu:\mathbb R^J\times \mathbb P^J\to\mathbb R_+$ is  such that for any $a>0$ and $(m,\Sigma)\in \mathbb R^J\times \mathbb P^J$, $S_\mu(am,\Sigma)=a^\chi\phi_\mu(V^{-1/2}m)$ for a $\mu$-smooth approximation $\phi_\mu$ of $\phi$.
\end{assumption}
In what follows, we also call  $S_\mu$ the $\mu$-smooth approximation of $S$.
Table \ref{tab:musmooth} gives the index functions we consider and their $\mu$-smooth approximations with associated parameters.
Details are provided in Appendix A. These functions satisfy Assumption \ref{as:onSmu}.\footnote{It is also possible to consider  the index function  $S(m,\Sigma)=\inf_{t\in\mathbb R^J_{-}}(m-t)'\Sigma^{-1}(m-t)$ considered in \cite{Rosen:2008aa} whose $\mu$-smooth approximation is given by $S_\mu(m,\Sigma)=\max_{u\in \mathbb R^J_+}m'u-\frac{1+2\mu}{4}u'\Sigma u$ with $\chi=2$. However, the statistic based on this approximation requires a second-order $\delta$-method to obtain a valid limiting distribution, which makes the analysis more complicated. As such, we leave this possibility for future research.\label{fnt:S2}}

\renewcommand{\arraystretch}{1.4}
\begin{table}[htbp]
\begin{center}
	\caption{$\mu$-smooth approximations of commonly used index functions.}
	\label{tab:musmooth}
\begin{tabular}{lllll}
\hline\hline	
	&  $S(m,\Sigma)$ & $S_\mu(m,\Sigma)$ & $(\alpha,\beta,K)$ \\
	  \hline
(i)	&  $\sum_{j=1}^J[m_j/\sigma_j]_+ $ & $\mu\sum_{j=1}^J\ln(\exp(\frac{m_j}{\mu\sigma_j})+1)$ & $(J,J\ln 2,0)$ \\
(ii)	&  $\max_{j=1,\cdots,J}\{m_j/\sigma_j\}_+$& $\mu\ln(\sum_{j=1}^J\exp(\frac{m_j}{\mu\sigma_j})+1)$ &$(1,\ln (J+1),0)$ \\
	\hline\hline
\end{tabular}
\end{center}
\end{table}

In summary, we propose the following procedure to construct confidence regions.

\bigskip
\noindent
\textbf{Algorithm 1:}
\begin{description}
	\item[Step 1]: Choose $\mu>0$. Calculate $\tilde T_{n,R}(\theta)$ using a $\mu$-smooth approximation $S_\mu$ of $S$ in Table \ref{tab:musmooth} and simulated samples of size $R$. For each observation $X_{i'}$, also draw a larger simulated  sample   $\{u_{i',1},\cdots,u_{i',R_2}\}$ with $R \ll R_2$ from the law $P(\cdot|X_{i'})$ for $i'=1,\cdots,n$.
	\item[Step 2]: Bootstrap.
	\begin{itemize}
		\item Let $\{X^*_1,\cdots,X^*_n\}$ be drawn from the empirical distribution of $\{X_1,\cdots, X_n\}$ (with replacement).
		\item For each bootstrapped observation $X^*_{i}$, draw a simulated sample $\{u_{i,1},\cdots,u_{i,R}\}$ from the law $P(\cdot|X^*_i)$ for $i=1,\cdots,n$.
		\item Compute the $1-\alpha$ quantile $c_{n,R,1-\alpha}(\theta)$ of the root:
		\begin{align}
			Z^*_{\mu,n,R,R_2}=\sqrt n(S_\mu(\bar m^*_{n,R},\hat\Sigma^*_{n,R})-S_\mu(\bar m_{n,R_2},\hat \Sigma_{n,R_2})).
		\end{align}
	\end{itemize}

	\item[Step 3]: Calculate the critical value $\tilde c_{n,R}(\theta)$ in \eqref{eq:crit_val} by introducing the correction term $\sqrt n\beta\mu$ for the approximation bias.
	\item[Step 4]: For each  $\theta\in\Theta$, conduct steps 1-3 and report $\mathcal{\tilde C}_{n,R}=\{\theta\in\Theta:\tilde T_{n,R}(\theta)\le \tilde c_{n,R}(\theta)\}.$
\end{description}
Similar to the consistent estimation of the asymptotic covariance matrix in the MSM literature,  we use a  simulation sample of larger size in the algorithm  \citep[see e.g.][Section 2.3]{GourierouxMontfort1996}.  This is to re-center the bootstrapped root at an object that tends to the population counterpart and  mimic the behavior of the root in \eqref{eq:def_zmunr} by the bootstrap sample. While $R_2$ needs to be large for the asymptotic approximation to be valid, this simulation needs to be done only once.

Below, we establish the asymptotic validity of the inference procedure described above.
For each $j$, we let $\sigma^2_{P,j}(\theta)\equiv Var_P(\hat m_{j,R}(X_i,u_i,\theta))$ and $\Omega_P(\theta)\equiv Corr_P(\hat m_{j,R}(X_i,u_i,\theta))$, where $\hat m_{j,R}(X_i,u_i,\theta)=R^{-1}\sum_{r=1}^R M_j(X_i,u_{i,r},\theta)$. We then let $D S_\mu$ denote the gradient of the map $m\mapsto S_\mu(m,\Sigma).$ Let $\Psi$ be the set of  $J$-by-$J$ correlation matrices $\Omega$ with $\det(\Omega)>\epsilon$ for some $\epsilon>0$.
 We make the following assumption on the model.

\begin{assumption}\label{as:onP}
	The model $\mathcal F$ for $(\theta,P)$ satisfies the following conditions.
	
	\noindent
	(i) $\theta\in\Theta;$
	
	\noindent
	(ii) $E_P[m_j(X_i,\theta)]\le 0,j=1,\cdots, J$ for some $\theta\in\Theta;$
	
	\noindent
	(iii) $\{X_i,u_{i,1},\cdots,u_{i,R},i=1,\cdots,n\}$ is an i.i.d. sample from $P$;
	
	\noindent
	(iv) $\sigma_{P,j}(\theta)\in (0,\infty)$ for $j=1,\cdots,J$ and $\theta\in\Theta$;
	
	\noindent
	(v) $\Omega_P(\theta)\in \Psi$;
	
	\noindent
	 (vi) $E_P\Big[\Big|\frac{m_j(X_i,\theta)}{\sigma_{P,j}(\theta)}\Big|^{2+\delta}\Big]\le M$ for some $\delta>0$ and $0<M<\infty$;
	
	 \noindent
	 (vii) For any $\mu\in[\underline M,\overline M]$, $\|DS_\mu(E_P[m(X_i,\theta)],\Sigma_P(\theta))\|\ge \eta$ for some $\eta>0$ and for all $\theta\in\Theta$.
\end{assumption}

Assumption \ref{as:onP} (i)-(vi) are based on the conditions in \cite{Andrews:2009aa,Andrews:2010jk}.
In (iii), we assume an i.i.d. sample. With this assumption, the following estimator of the asymptotic covariance can be used:\footnote{While we establish validity of inference for i.i.d. samples, one could potentially relax this assumption and allow for strictly stationary and strongly mixing data by adopting an alternative estimator of the asymptotic covariance and modifying the bootstrap procedure properly.  See \cite{Andrews:2009aa} for the inference framework that allows such data.}
\begin{align}
	\hat\Sigma_{n,R}(\theta)=\frac{1}{n}\sum_{i=1}^n (\hat m_R(X_i,\theta)-\bar m_{n,R}(\theta))(\hat m_R(X_i,\theta)-\bar m_{n,R}(\theta))'.\label{eq:hatSigma}
\end{align}
We then let $\hat V_{n,R}(\theta)=\text{diag}(\hat\Sigma_{n,R}(\theta))$. For the purpose of stating an assumption, let us also define an infeasible estimator of $\Sigma_P$ as follows:
\begin{align*}
\hat\Sigma_n(\theta)=\frac{1}{n}\sum_{i=1}^n (m(X_i,\theta)-\bar m_{n}(\theta))(m(X_i,\theta)-\bar m_{n}(\theta))'.
\end{align*}
This estimator can be computed only if simulation is not required. We then let $\hat V_{n}(\theta)=\text{diag}(\hat\Sigma_{n}(\theta))$.

In Assumption \ref{as:onP} (v),  $\Psi$ contains all $J\times J$ correlation matrices whose determinant is bounded from below by $\epsilon>0$. This condition is required for one of the index functions used in \cite{Andrews:2010jk}. In our setting, we use this condition to ensure that the limiting distribution of $Z_{\mu,n,R}$ is continuously distributed. With additional notation, this assumption can be relaxed so that the lower bound on the determinant is required  only for the correlation matrix of a suitable subset of the moment functions.\footnote{See \cite{Kaido:2017aa} for generalization of the condition along this line.} Or it can be dropped entirely at the price of an additional tuning parameter to handle a potentially discontinuos limiting distribution.

Assumption \ref{as:onP} (vii) is an additional condition, which we add to the standard set of assumptions. It requires that the gradient of the smoothed index function $S_\mu$ does not vanish. This allows us to use the (first-order) $\delta$-method.  For the functions in Table \ref{tab:musmooth}, the condition is satisfied when $\theta\mapsto E_P[m_j(X_i,\theta)]$ is continuous and $\sigma_{P,j}(\theta),j=1,\cdots,J$ are uniformly bounded away from 0.

The following theorem ensures that the proposed confidence region controls the asymptotic confidence size uniformly over the parameter space $\mathcal F$.
\begin{theorem}\label{thm:size_control}
	Suppose Assumptions \ref{as:onSmu}-\ref{as:onP} hold. Let $0<\underline M<\overline M<\infty$.
	Let $R\in\mathbb N$ be fixed and let $\{R_2\}\subset \mathbb N$ be a sequence  such that $R_2\to \infty$ as $n\to\infty$ and $\hat V_{n,R_2}(\theta_n)^{-1/2}\bar m_{n,R_2}(\theta_n)-\hat V_{n}(\theta_n)^{-1/2}\bar m_{n}(\theta_n)=o_P(n^{-1/2})$ uniformly in $P$. Then,
\begin{align*}
\liminf_{n\to\infty}\inf_{\mu\in[\underline M,\overline M]}\inf_{(\theta,P)\in\mathcal F}P\Big(\theta\in\tilde{\mathcal C}_{n,R}\Big)\ge 1-\alpha.
\end{align*}	
\end{theorem}

\section{Monte Carlo Experiments}\label{sec:montecarlo}
In this section, we show results of main Monte Carlo simulations to examine the performance of the methods that account for the simulation error.

\subsection{Performance of methods with correction for simulation errors}
Following our example on the intersection bounds in Section \ref{ssec:num_ex1},  we first use the method that corrects the critical value described in Section \ref{ssec:cv_correction}. We construct confidence intervals for $\theta $  with level $1-\alpha$ using:
\begin{align}
	\mathcal C_{n,R}^{\text{CV}}&\equiv\big(-\infty,~\min_{j}\{\bar{m}_{j,n,R}+ c_{\kappa,n,R}\frac{\hat{\sigma}_{j,n,R}}{\sqrt{n}}\big\}\big],\label{eq:clr_sim}
\end{align}
where $\bar{m}_{j,n,R}=(nR)^{-1}\sum_{i=1}^{n}\sum_{r=1}^{R}1\{u_{j,i,r}<X_{j,i}\},$ $\hat{\sigma}_{j,n,R}$ is the estimated standard deviation of the $j$-th simulated moment, and $c_{\kappa,n,R}$ is a critical value computed  as the $1-\alpha$ quantile of
\begin{align}
T^*_{n,R}= \max_{j\in\{j:\xi_{j,n,R}\ge -1\}}\sqrt{n}(\bar{m}_{j,n,R}^{*}-\bar{m}_{j,n,R})/\hat{\sigma}_{j,n,R}^{*},
 \end{align}
where  $(\bar{m}_{j,n,R}^{*},\hat{\sigma}_{j,n,R}^{*})$ is the bootstrap quantities corresponding to $(\bar{m}_{j,n,R},\hat{\sigma}_{j,n,R})$, and $\xi_{j,n,R}$ is as defined in \eqref{eq:xis} with $\kappa_n=\sqrt{\ln n}$.  When $T^*_{n,R}$ is computed, we redraw simulation draws  $\{u_{i,r},r=1,...R\}$ to account for simulation variations when simulating the bootstrap samples.

We also consider our  inference procedure based on the regularized statistic. Specifically, we use $\phi_{\mu}(\bar{m}_{j,n,R})=-\mu \ln\sum_{j=1}^{J}e^{-\bar{m}_{j,n,R}/\mu}$ to approximate $\phi(\bar{m}_{j,n,R})=\min_{j=1,..J}\bar{m}_{j,n,R}$ used in computing the confidence region in equation~\eqref{eq:cs_sim}.
We then construct  confidence intervals for $\theta$   with level $1-\alpha$ using:
\begin{align}
	\tilde{\mathcal C}_{n,R}&\equiv\big(-\infty,\phi_{\mu}(\bar{m}_{j,n,R})+\tilde{c}_{\mu,n,R}/\sqrt{n}\big],
\end{align}
where  $\tilde{c}_{n,R}$ is an estimated critical value in \eqref{eq:crit_val}.  We also redraw simulation draws  $\{u_{i,r},r=1,...R\}$ when simulating the bootstrap samples. In the experiments, we use $R_2=100$  to compute $(\bar m_{n,R_2},\hat\Sigma_{n,R_2})$ in the bootstrap step (Step 2 in Algorithm 1). For both confidence intervals, $S=1000$ Monte Carlo replications are generated.

Table~\ref{tab:CLR_bind} reports  the probabilities of the confidence intervals covering the upper bound $\theta^U$ of the identified set using the critical value correction. The coverage probabilities of the confidence intervals are all above the nominal level after correcting the critical value. This can be contrasted with the coverage probabilities of the confidence intervals without any correction.
The critical value correction method tends to make the confidence interval somewhat conservative when $R$ is small. In some cases, the coverage probability is very close to 1 and the corrected confidence interval is substantially longer than the one based on the regularization method as we discuss below.

Table~\ref{tab:mu24_bind} reports  the coverage probabilities of the confidence intervals based on  the $\mu$-smooth approximation. We set the smoothing parameter to $\mu=0.02$ and $0.04$ in our Monte Carlo experiments. The coverage probabilities of $\tilde{\mathcal C}_{n,R}$ are  all above those of the confidence intervals without correction and close to the nominal level in many cases. 
 These results indicate that  the  inference procedure with the $\mu$-smooth approximation is an effective method for correcting the size distortion caused by the finite number of draws.

 Tables~\ref{tab:length_CLR_bind} and \ref{tab:lengthmu24_bind} report  the median length of the  confidence intervals. The median length of a one-sided confidence interval  is computed as the difference between the median of the upper bound of the confidence interval and the right end point of the identified set. Table \ref{tab:length_CLR_bind}  shows that the confidence interval with the critical value correction is often much longer than
 the confidence interval without any correction, which is consistent with the robust size control property. The difference between the two shrinks as $R$ gets large.
Comparing the two tables, one can see that  the regularization based confidence intervals $\tilde{\mathcal C}_{n,R}$  are often significantly shorter than $\mathcal C_{n,R}^{\text{CV}}$.  A close inspection of the simulation results showed that this was because the regularized statistic  had a smaller variance that that of the non-regularized statistic, which led to shorter confidence intervals even after the bias correction.\footnote{This may be a generic feature of the $\mu$-smooth approximation method. We leave its general analysis for future work.}

The overall pattern remains the same even in the presence of locally slack constraints. Tables~\ref{tab:CLR_slack}-\ref{tab:lengthmu24_slack} report the coverage probabilities and excess length of the confidence intervals
when some of the constraints are slack. The slack constraints are introduced by shifting $X_j$'s mean by $1/\sqrt n$ for the first $J/5$ constraints.
One can see that both correction methods achieve valid coverage across all values of simulation draws. The regularization based confidence interval becomes slightly more conservative in terms of coverage probabilities compared to the case without the slack constraints. However, its length is still shorter than that of the critical value correction confidence interval across all cases.

In sum, the simulation results show that the regularization based confidence interval $\tilde{\mathcal C}_{n,R}$  works well both in terms of size  and length. Its size is controlled reasonably well even under some DGPs that make simulated confidence intervals without  correction severely undersized. The critical value correction method also achieves robust size control. However, it tends to be overly conservative when $R$ is small.

\section{Concluding remarks}
This paper explores the effects of simulated moments on the performance of
inference methods based on moment inequalities.  Due to the irregularity of the boundary
of the confidence regions, simulation errors can affect the performance of
inference in non-standard ways. This can result in a severe distortion especially when the number of inequality restrictions is large
and the essential sample size is small.
To account for the effect of the simulation error, we propose a novel way to construct confidence regions
using regularized statistics and establish an asymptotic size control result. 
The simulation results confirm the robust size control property of the proposed method.
An interesting avenue for future research  is on the choice of the smoothing parameter that accounts for the trades-off between the amount of bias correction and variance of the regularized statistic.

\clearpage
\begin{table}[h]
	\small
\centering
\caption{Coverage Probabilities for True Upper Bound: Critical Value Correction}
\label{tab:CLR_bind}
\begin{tabular}{lccccccccc}
\hline\hline
 & \multicolumn{4}{c}{Simulated (No correction)} && \multicolumn{4}{c}{CV correction: $\mathcal C_{n,R}^{\text{CV}}$}\tabularnewline
\cline{2-5}
\cline{7-10}
 & $R=1$ & $R=5$ & $R=10$ & $R=20$ && $R=1$ & $R=5$ & $R=10$ & $R=20$\tabularnewline
\hline
\multicolumn{1}{l}{A: ($J=5$)} & & & & & & & & \\
$n=100 $  & 0.583 & 0.882 & 0.908 & 0.930 && $0.998$ & $0.977$ & $0.959$ & $0.954$\tabularnewline
$n=250 $  & 0.604 & 0.881 & 0.920 & 0.935 && $0.998$ & $0.969$ & $0.968$ & $0.965$\tabularnewline
$n=1000$  & 0.666 & 0.894 & 0.913 & 0.933 && $0.997$ & $0.983$ & $0.968$ & $0.957$\tabularnewline
 & & & & & & & & & \\
\multicolumn{1}{l}{B: ($J=10$)}  & & & & & & & & \\
$n=100 $  & 0.481 & 0.853 & 0.904 & 0.920 && $0.996$ & $0.972$ & $0.955$ & $0.948$\tabularnewline
$n=250 $  & 0.467 & 0.868 & 0.908 & 0.925 && $0.998$ & $0.980$ & $0.964$ & $0.957$\tabularnewline
$n=1000$  & 0.487 & 0.853 & 0.896 & 0.919 && $1.000$ & $0.983$ & $0.963$ & $0.951$\tabularnewline
 & & & & & & & & & \\
\multicolumn{1}{l}{C: ($J=30$)}  & & & & & & & & \\
$n=100 $  & 0.245 & 0.811 & 0.888 & 0.923 && $0.996$ & $0.982$ & $0.971$ & $0.960$\tabularnewline
$n=250 $  & 0.266 & 0.803 & 0.878 & 0.917 && $0.999$ & $0.982$ & $0.970$ & $0.958$\tabularnewline
$n=1000$  & 0.235 & 0.810 & 0.881 & 0.912 && $1.000$ & $0.983$ & $0.970$ & $0.956$\tabularnewline
\hline\hline
\end{tabular}

\vspace{0.2in}
{\centering
\caption{Coverage Probabilities for True Upper Bound: $\mu$-smooth Approximation}
\label{tab:mu24_bind}
\begin{tabular}{lccccccccc}
\hline\hline
 & \multicolumn{4}{c}{Regularization:  ($\mu=0.02$)} && \multicolumn{4}{c}{Regularization: ($\mu=0.04$)}\tabularnewline
\cline{2-5}
\cline{7-10}
 & $R=1$ & $R=5$ & $R=10$ & $R=20$ && $R=1$ & $R=5$ & $R=10$ & $R=20$\tabularnewline
\hline
\hline
\multicolumn{1}{l}{A: ($J=5$)} & & & & & & & & \\
$n=$100   & $0.940$ & $0.959$ & $0.953$ & $0.961$ && $0.957$ & $0.962$ & $0.955$ & $0.961$\tabularnewline
$n=$250   & $0.961$ & $0.952$ & $0.955$ & $0.963$ && $0.962$ & $0.948$ & $0.943$ & $0.958$\tabularnewline
$n=$1000  & $0.952$ & $0.962$ & $0.958$ & $0.947$ && $0.950$ & $0.953$ & $0.958$ & $0.943$\tabularnewline
 & & & & & & & & & \\
\multicolumn{1}{l}{B: ($J=10$)}  & & & & & & & & \\
$n=$100   & $0.946$ & $0.964$ & $0.958$ & $0.969$ && $0.952$ & $0.963$ & $0.959$ & $0.966$\tabularnewline
$n=$250   & $0.949$ & $0.962$ & $0.972$ & $0.959$ && $0.949$ & $0.960$ & $0.956$ & $0.953$\tabularnewline
$n=$1000  & $0.955$ & $0.949$ & $0.955$ & $0.954$ && $0.958$ & $0.947$ & $0.952$ & $0.949$\tabularnewline
 & & & & & & & & & \\
\multicolumn{1}{l}{C: ($J=30$)}  & & & & & & & & \\
$n=$100   & $0.964$ & $0.974$ & $0.979$ & $0.976$ && $0.964$ & $0.968$ & $0.966$ & $0.968$\tabularnewline
$n=$250   & $0.970$ & $0.968$ & $0.968$ & $0.971$ && $0.962$ & $0.956$ & $0.959$ & $0.963$\tabularnewline
$n=$1000  & $0.950$ & $0.954$ & $0.957$ & $0.948$ && $0.947$ & $0.952$ & $0.957$ & $0.944$\tabularnewline
\hline\hline
\end{tabular}}
\end{table}

\begin{table}[h]
	\small
\centering
\caption{Excess Lengths of Confidence Intervals: Critical Value Correction}
\label{tab:length_CLR_bind}
\begin{tabular}{lccccccccc}
\hline\hline
 & \multicolumn{4}{c}{Simulated (No correction)} && \multicolumn{4}{c}{CV correction: $\mathcal C_{n,R}^{\text{CV}}$}\tabularnewline
\cline{2-5}
\cline{7-10}
  & $R=1$ & $R=5$ & $R=10$ & $R=20$ && $R=1$ & $R=5$ & $R=10$ & $R=20$\tabularnewline
\hline
\multicolumn{1}{l}{A: ($J=5$)} & & & & & & & & \\
$n=$100   & $0.007$ & $0.031$ & $0.031$ & $0.034$ && $0.081$ & $0.049$ & $0.044$ & $0.041$\tabularnewline
$n=$250   & $0.007$ & $0.019$ & $0.021$ & $0.021$ && $0.052$ & $0.031$ & $0.027$ & $0.024$\tabularnewline
$n=$1000  & $0.004$ & $0.009$ & $0.010$ & $0.011$ && $0.027$ & $0.016$ & $0.014$ & $0.013$\tabularnewline
 & & & & & & & & & \\
\multicolumn{1}{l}{B: ($J=10$)}  & & & & & & & & \\
$n=$100   & $-0.006$ & $0.022$ & $0.026$ & $0.028$ && $0.083$ & $0.047$ & $0.040$ & $0.036$\tabularnewline
$n=$250   & $-0.001$ & $0.015$ & $0.017$ & $0.018$ && $0.053$ & $0.029$ & $0.025$ & $0.022$\tabularnewline
$n=$1000  & $-0.001$ & $0.007$ & $0.009$ & $0.009$ && $0.027$ & $0.015$ & $0.013$ & $0.011$\tabularnewline
 & & & & & & & & & \\
\multicolumn{1}{l}{C: ($J=30$)}  & & & & & & & & \\
$n=$100   & $-0.015$ & $0.017$ & $0.022$ & $0.024$ && $0.085$ & $0.045$ & $0.038$ & $0.032$\tabularnewline
$n=$250   & $-0.010$ & $0.010$ & $0.013$ & $0.015$ && $0.053$ & $0.028$ & $0.023$ & $0.020$\tabularnewline
$n=$1000  & $-0.005$ & $0.005$ & $0.007$ & $0.008$ && $0.027$ & $0.014$ & $0.012$ & $0.010$\tabularnewline
\hline\hline
\end{tabular}

\vspace{0.2in}
{\centering
\caption{Excess Lengths of Confidence Intervals: $\mu$-smooth Approximation}
\label{tab:lengthmu24_bind}
\begin{tabular}{lccccccccc}
\hline\hline
 & \multicolumn{4}{c}{Regularization:  ($\mu=0.02$)} & \multicolumn{4}{c}{Regularization: ($\mu=0.04$)}\tabularnewline
\cline{2-5}
\cline{7-10}
  & $R=1$ & $R=5$ & $R=10$ & $R=20$ & & $R=1$ & $R=5$ & $R=10$ & $R=20$\tabularnewline
\hline
\multicolumn{1}{l}{A: ($J=5$)} & & & & & & & & \\
$n=$100   & $0.050$ & $0.034$ & $0.030$ & $0.029$ && $0.047$ & $0.032$ & $0.029$ & $0.027$\tabularnewline
$n=$250   & $0.030$ & $0.019$ & $0.017$ & $0.017$ && $0.028$ & $0.018$ & $0.016$ & $0.016$\tabularnewline
$n=$1000  & $0.013$ & $0.009$ & $0.008$ & $0.008$ && $0.012$ & $0.008$ & $0.008$ & $0.008$\tabularnewline
 & & & & & & & & & \\
\multicolumn{1}{l}{B: ($J=10$)}  & & & & & & & & \\
$n=$100   & $0.040$ & $0.025$ & $0.023$ & $0.022$ && $0.037$ & $0.022$ & $0.021$ & $0.020$\tabularnewline
$n=$250   & $0.022$ & $0.014$ & $0.015$ & $0.012$ && $0.022$ & $0.013$ & $0.012$ & $0.011$\tabularnewline
$n=$1000  & $0.009$ & $0.006$ & $0.006$ & $0.006$ && $0.009$ & $0.006$ & $0.006$ & $0.005$\tabularnewline
 & & & & & & & & & \\
\multicolumn{1}{l}{C: ($J=30$)}  & & & & & & & & \\
$n=$100   & $0.028$ & $0.017$ & $0.016$ & $0.014$ && $0.023$ & $0.015$ & $0.013$ & $0.012$\tabularnewline
$n=$250   & $0.015$ & $0.009$ & $0.008$ & $0.008$ && $0.012$ & $0.008$ & $0.007$ & $0.007$\tabularnewline
$n=$1000  & $0.005$ & $0.003$ & $0.003$ & $0.003$ && $0.005$ & $0.003$ & $0.003$ & $0.003$\tabularnewline
\hline\hline
\end{tabular}}
\end{table}

\begin{table}[h]
	\small
\centering
\caption{Coverage Probabilities for True Upper Bound: Critical Value Correction}
\label{tab:CLR_slack}
\begin{tabular}{lccccccccc}
\hline\hline
 & \multicolumn{4}{c}{Simulated (No correction)} & & \multicolumn{4}{c}{CV correction: $\mathcal C_{n,R}^{\text{CV}}$}\tabularnewline
\cline{2-5}
\cline{7-10}
 & $R=1$ & $R=5$ & $R=10$ & $R=20$ & & $R=1$ & $R=5$ & $R=10$ & $R=20$\tabularnewline
\hline
\multicolumn{1}{l}{A: ($J=5$)} & & & & & & & & \\
$n=100 $  & 0.620 & 0.899 & 0.920 & 0.935 && $0.998$ & $0.980$ & $0.966$ & $0.960$\tabularnewline
$n=250 $  & 0.647 & 0.903 & 0.933 & 0.945 && $0.999$ & $0.981$ & $0.970$ & $0.969$\tabularnewline
$n=1000$  & 0.700 & 0.911 & 0.927 & 0.944 && $0.998$ & $0.984$ & $0.972$ & $0.967$\tabularnewline
 & & & & & & & & & \\
\multicolumn{1}{l}{B: ($J=10$)}  & & & & & & & & \\
$n=100 $  & 0.533 & 0.875 & 0.920 & 0.938 && $0.997$ & $0.975$ & $0.963$ & $0.959$\tabularnewline
$n=250 $  & 0.520 & 0.886 & 0.922 & 0.935 && $0.999$ & $0.983$ & $0.966$ & $0.960$\tabularnewline
$n=1000$  & 0.535 & 0.879 & 0.917 & 0.933 && $1.000$ & $0.988$ & $0.973$ & $0.960$\tabularnewline
 & & & & & & & & & \\
\multicolumn{1}{l}{C: ($J=30$)}  & & & & & & & & \\
$n=100 $  & 0.301 & 0.833 & 0.911 & 0.940 && $0.998$ & $0.985$ & $0.976$ & $0.967$\tabularnewline
$n=250 $  & 0.324 & 0.840 & 0.895 & 0.933 && $0.999$ & $0.984$ & $0.974$ & $0.964$\tabularnewline
$n=1000$  & 0.280 & 0.833 & 0.905 & 0.926 && $1.000$ & $0.986$ & $0.975$ & $0.961$\tabularnewline
\hline\hline
\end{tabular}

\vspace{0.2in}
{\centering
\caption{Coverage Probabilities for True Upper Bound: $\mu$-smooth Approximation}
\label{tab:mu24_slack}
\begin{tabular}{lccccccccc}
\hline\hline
 & \multicolumn{4}{c}{Regularization:  ($\mu=0.02$)} & & \multicolumn{4}{c}{Regularization: ($\mu=0.04$)}\tabularnewline
\cline{2-5}
\cline{7-10}
 & $R=1$ & $R=5$ & $R=10$ & $R=20$ & & $R=1$ & $R=5$ & $R=10$ & $R=20$\tabularnewline
\hline
\hline
\multicolumn{1}{l}{A: ($J=5$)} & & & & & & & & \\
$n=$100  & $0.948$ & $0.957$ & $0.956$ & $0.958$ && $0.965$ & $0.977$ & $0.969$ & $0.973$\tabularnewline
$n=$250  & $0.963$ & $0.972$ & $0.977$ & $0.985$ && $0.973$ & $0.970$ & $0.976$ & $0.985$\tabularnewline
$n=$1000 & $0.968$ & $0.984$ & $0.977$ & $0.979$ && $0.969$ & $0.980$ & $0.976$ & $0.979$\tabularnewline
 & & & & & & & & & \\
\multicolumn{1}{l}{B: ($J=10$)}  & & & & & & & & \\
$n=$100  & $0.934$ & $0.969$ & $0.963$ & $0.973$ && $0.978$ & $0.982$ & $0.979$ & $0.991$\tabularnewline
$n=$250  & $0.968$ & $0.984$ & $0.986$ & $0.984$ && $0.978$ & $0.983$ & $0.981$ & $0.984$\tabularnewline
$n=$1000 & $0.979$ & $0.983$ & $0.983$ & $0.985$ && $0.978$ & $0.983$ & $0.983$ & $0.982$\tabularnewline
 & & & & & & & & & \\
\multicolumn{1}{l}{C: ($J=30$)}  & & & & & & & & \\
$n=$100  & $0.962$ & $0.976$ & $0.990$ & $0.986$ && $0.980$ & $0.991$ & $0.993$ & $0.993$\tabularnewline
$n=$250  & $0.982$ & $0.993$ & $0.997$ & $0.996$ && $0.985$ & $0.993$ & $0.994$ & $0.995$\tabularnewline
$n=$1000 & $0.979$ & $0.993$ & $0.994$ & $0.995$ && $0.980$ & $0.994$ & $0.993$ & $0.996$\tabularnewline
\hline\hline
\end{tabular}}

\vspace{0.1in}

Note: Among $J$ moment inequalities, $J/5$ of them are slack.
\end{table}

\begin{table}[h]
	\small
\centering
\caption{Excess Lengths of Confidence Intervals: Critical Value Correction }
\label{tab:length_CLR_slack}
\begin{tabular}{lccccccccc}
\hline\hline
  & \multicolumn{4}{c}{Simulated (No correction)} & & \multicolumn{4}{c}{CV correction: $\mathcal C_{n,R}^{\text{CV}}$}\tabularnewline
\cline{2-5}
\cline{7-10}
 &   $R=1$ & $R=5$ & $R=10$ & $R=20$ & & $R=1$ & $R=5$ & $R=10$ & $R=20$\tabularnewline
\hline
\multicolumn{1}{l}{A: ($J=5$)} & & & & & & & & \\
$n=$100   & $0.017$ & $0.033$ & $0.034$ & $0.037$ && $0.091$ & $0.052$ & $0.047$ & $0.044$\tabularnewline
$n=$250   & $0.011$ & $0.020$ & $0.022$ & $0.023$ && $0.056$ & $0.034$ & $0.029$ & $0.027$\tabularnewline
$n=$1000  & $0.005$ & $0.011$ & $0.011$ & $0.012$ && $0.028$ & $0.017$ & $0.015$ & $0.014$\tabularnewline
 & & & & & & & & & \\
\multicolumn{1}{l}{B: ($J=10$)}  & & & & & & & & \\
$n=$100  & $0.004$ & $0.026$ & $0.030$ & $0.031$ && $0.091$ & $0.051$ & $0.044$ & $0.039$\tabularnewline
$n=$250  & $0.003$ & $0.016$ & $0.019$ & $0.020$ && $0.056$ & $0.032$ & $0.027$ & $0.024$\tabularnewline
$n=$1000 & $0.001$ & $0.008$ & $0.010$ & $0.010$ && $0.028$ & $0.016$ & $0.014$ & $0.012$\tabularnewline
 & & & & & & & & & \\
\multicolumn{1}{l}{C: ($J=30$)}  & & & & & & & & \\
$n=$100  & $-0.015$ & $0.019$ & $0.025$ & $0.027$ && $0.091$ & $0.048$ & $0.041$ & $0.035$\tabularnewline
$n=$250  & $-0.006$ & $0.012$ & $0.015$ & $0.017$ && $0.056$ & $0.030$ & $0.025$ & $0.022$\tabularnewline
$n=$1000 & $-0.004$ & $0.006$ & $0.008$ & $0.009$ && $0.027$ & $0.015$ & $0.013$ & $0.011$\tabularnewline
\hline\hline
\end{tabular}

\vspace{0.2in}
{\centering
\caption{Excess Lengths of Confidence Intervals: $\mu$-smooth Approximation}
\label{tab:lengthmu24_slack}
\begin{tabular}{lccccccccc}
\hline\hline
 & \multicolumn{4}{c}{Regularization:  ($\mu=0.02$)} & & \multicolumn{4}{c}{Regularization: ($\mu=0.04$)}\tabularnewline
\cline{2-5}
\cline{7-10}
   & $R=1$ & $R=5$ & $R=10$ & $R=20$ & & $R=1$ & $R=5$ & $R=10$ & $R=20$\tabularnewline
\hline
\multicolumn{1}{l}{A: ($J=5$)} & & & & & & & & & \\
$n=$100   & $0.054$ & $0.040$ & $0.036$ & $0.036$ & & $0.051$ & $0.036$ & $0.033$ & $0.032$\tabularnewline
$n=$250   & $0.035$ & $0.024$ & $0.022$ & $0.021$ & & $0.031$ & $0.021$ & $0.020$ & $0.019$\tabularnewline
$n=$1000  & $0.016$ & $0.011$ & $0.010$ & $0.010$ & & $0.014$ & $0.010$ & $0.010$ & $0.009$\tabularnewline
  & & & & & & & & & \\
\multicolumn{1}{l}{B: ($J=10$)}  & & & & & & & & \\
$n=$100  & $0.047$ & $0.031$ & $0.030$ & $0.029$ & &$0.042$ & $0.027$ & $0.025$ & $0.025$\tabularnewline
$n=$250  & $0.028$ & $0.019$ & $0.018$ & $0.017$ & &$0.023$ & $0.016$ & $0.015$ & $0.014$\tabularnewline
$n=$1000 & $0.012$ & $0.008$ & $0.008$ & $0.008$ & &$0.010$ & $0.008$ & $0.007$ & $0.007$\tabularnewline
  & & & & & & & & & \\
\multicolumn{1}{l}{C: ($J=30$)}  & & & & & & & & & \\
$n=$100   & $0.036$ & $0.025$ & $0.024$ & $0.022$ & & $0.028$ & $0.019$ & $0.018$ & $0.017$\tabularnewline
$n=$250   & $0.021$ & $0.014$ & $0.013$ & $0.012$ & & $0.015$ & $0.011$ & $0.010$ & $0.010$\tabularnewline
$n=$1000  & $0.008$ & $0.005$ & $0.005$ & $0.005$ & & $0.007$ & $0.005$ & $0.005$ & $0.005$\tabularnewline
\hline\hline
\end{tabular}}

\vspace{0.1in}

Note: Among $J$ moment inequalities, $J/5$ of them are slack.
\end{table}

\clearpage

\begin{small}

\bibliographystyle{newapa}
\end{small}

\clearpage

\appendix

\section{Proofs}
\subsection{Hausdorff consistency}
Below, we let $\xi=(X,u_1,\cdots,u_R)\in \mathcal X\times \mathcal U^R$ be a random vector that stacks $X$ and $(u_1,\cdots,u_R)$.
We then let $\hat m_R(\xi,\theta)\equiv R^{-1}\sum_{r=1}^R M(X,u_r,\theta)$.
For each $\theta$, let $\hat\Sigma_{n,R}(\theta)$ be defined as in \eqref{eq:hatSigma} and let $\hat V_{n,R}(\theta)=\text{diag}(\hat\Sigma_{n,R}(\theta))$.
We use the following assumption to establish consistency of $\hat\Theta_n(c_{n,R})$ in Proposition \ref{prop:consistency}.
\begin{assumption}\label{as:cht}
The following conditions hold.
\begin{itemize}
	\item[(i)] The parameter space $\Theta$ is a nonempty compact subset of $\mathbb R^d$;
	\item[(ii)] $M:\mathcal X\times \mathcal U\times\Theta'\to\mathbb R^J$: is jointly measurable in $(x,u,\theta)$, where $\Theta'$ is a neighborhood of $\Theta$;
	\item[(iii)]  For each $R$, the collection $\{\hat m_R(\cdot,\theta),\theta\in\Theta'\}$ is a $P$-Donsker class. $\{\xi_i\}_{i=1}^n$ is an i.i.d. sample;
	\item[(iv)] There exist positive constants $C$ and $\delta$ such that for all $\theta\in\Theta$,
	\begin{align*}
	\|E_P[m(X_i,\theta)]\|_+\ge C(d(\theta,\Theta_I)\wedge \delta).
	\end{align*}
	\item[(v)] $\hat V_{n,R}(\theta)^{-1}=V_P(\theta)^{-1}+o_p(1)$ uniformly in $\theta\in\Theta'$ when $n\to\infty$ with $R$ fixed, where $V_P(\theta)=diag(\Sigma_P(\theta))$ is a diagonal matrix with positive diagonal elements and is continuous for all $\theta\in\Theta'.$
\end{itemize}	
\end{assumption}
The imposed conditions are analogous to Condition M.1 in \cite{ChernozhukovHanTamer2007} (see Section 3.2 in their paper for details). A key modification is that we assume that $\{\hat m_R(\cdot,\theta),\theta\in\Theta\{$ is $P$-Donsker instead of $\{ m(\cdot,\theta),\theta\in\Theta\}$ being $P$-Donsker. This ensures that the sample moments converge to the population counterpart even with a finite number of draws. This condition is satisfied when $\{M(\cdot,\cdot,\theta),\theta\in\Theta\}$ is a $P$-Donsker class, which holds for a wide class of functions \citep{PakesPollard1989,Van-Der-Vaart:1996aa}. Assumption \ref{as:cht} (v) requires that the estimators of the asymptotic variance converge with a finite number of draws. This holds under mild  regularity conditions. A more primitive condition is given in Assumption \ref{as:onP}.

\bigskip
\begin{proof}[\rm Proof of Proposition \ref{prop:consistency}]
The result follows almost immediately from Theorems 3.1 and 4.2 of \cite{ChernozhukovHanTamer2007} (CHT below). Hence, we briefly sketch the argument below.

Under the imposed assumptions, Condition C.1 (a)-(c) in CHT follows. This step is the same as in the proof of Theorem 4.2 in CHT.
For Condition C.1-(d), note that
\begin{align*}
E_\xi[\hat m_R(\xi,\theta)]=E_\xi[R^{-1}\sum_{r=1}^RM(X,u_{r},\theta)]=E_X[R^{-1}\sum_{r=1}^R E_{u_r}[M(X,u_{r},\theta)|X]]=E_P[m(X,\theta)].
\end{align*}
By  $\{\hat m_R(\cdot,\theta),\theta\in\Theta\}$ being $P$-Donsker, it follows that
\begin{align*}
\sup_{\theta\in\Theta}\big\|n^{-1}\sum_{i=1}^n\hat m_R(\xi_i,\theta)-E_P[m(X_i,\theta)]\big\|=o_p(1),
\end{align*}
for a fixed $R$ and $n\to\infty$. Together with the uniform convergence of the weighting matrix (Assumption \ref{as:cht} (v)) and the continuous mapping theorem, it implies Condition C.1-(d) of CHT.

Again by the Donskerness,  uniformly over $\Theta_I$
\begin{align*}
n\|\bar m_{n,R}(\theta)'\hat V_{n,R}(\theta)^{-1/2}\|^2_+=\|\sqrt n(\bar  m_{n,R}(\theta)-E[m(X_i,\theta)]+\xi_n(\theta))]'\hat V_{n,R}(\theta)^{-1/2}\|_+^2=O_p(1),
\end{align*}
where the last equality follows because of the Donskerness and $\xi_n(\theta)\le 0$ on $\Theta_I$. Hence, one may take $a_n=n$. The conclusion of the proposition then follows from Theorem 3.1 in CHT.	
\end{proof}

\subsection{Auxiliary results for regularization based confidence regions}
For each $\theta\in\Theta$ and $\mu>0$, define
\begin{align}
Z_{\mu,n,R}(\theta)&\equiv S_\mu(\sqrt n\bar m_{n,R}(\theta),\hat\Sigma_{n,R}(\theta))-S_\mu(\sqrt nE_P[m(X_i,\theta)],\Sigma_P(\theta))\label{eq:B1}\\
b_\mu(\theta)&\equiv S_\mu(\sqrt nE[m(X_i,\theta)],\Sigma_P(\theta))-S(\sqrt nE_P[m(X_i,\theta)],\Sigma_P(\theta)).\label{eq:B2}
\end{align}	
In what follows, we denote the covariance matrix of the moment function under distribution $P$ by $\Sigma_P(\theta)=E_P[m(X_i,\theta)m(X_i,\theta)'],$ and we let $V_P(\theta)=\text{diag}(\Sigma_P(\theta))$. Similarly, we let the correlation matrix be $\Omega_P(\theta)=V_P(\theta)^{-1/2}\Sigma_P(\theta)V_P(\theta)^{-1/2}$.
Here is a set of high-level conditions for the asymptotic size control.

\begin{condition}\label{cond:high1}
There exist $\mathcal M\subset \mathbb R_{++}$, $\mathcal C\subset\mathbb R$, and $\mathcal F\subset \Theta\times \mathcal P$ such that $\mathcal M,\mathcal C$ are  closed intervals, and the following conditions hold:

\noindent (i) If a sequence $(\mu_n,c_n,\theta_n,P_n)\in\mathcal M\times\mathcal C\times\mathcal F$ satisfies $E_{P_n}[m(X_i,\theta_n)]\to m^*$, $\Omega_{P_n}(\theta_n)\to \Omega,$ and $(\mu_n,c_n)\to (\mu,c)\in \mathcal M\times\mathcal C$, then
\begin{align}
P_n(Z_{\mu_n,n,R}(\theta_n)\le c_n)\to \text{Pr}( Z_{\mu,R}\le c)\label{eq:high1_1}
\end{align}
for all continuity points of $Z_{\mu,R}$ whose distribution is determined by $(\Omega,\mu,R)$. The cumulative distribution function $J(c)=Pr(Z_{\mu,R}\le c)$ is strictly increasing at the $1-\alpha$ quantile.

\noindent (ii)
for the sequence $(\mu_n,\theta_n,P_n)$ defined above,
\begin{align}
c_{n,R,1-\alpha}(\theta_n)\stackrel{P_n}{\to}c_{\mu,1-\alpha},\label{eq:high1_2}
\end{align}
where $c_\mu$ is the $1-\alpha$ quantile of $Z_{\mu,R}$.
\end{condition}

Under this condition, we obtain the following generic size control result, which we will apply to establish Theorem \ref{thm:size_control}.
\begin{proposition}\label{prop:size}
	Suppose that $S_\mu$ is a $\mu$-smooth approximation to $S$. Suppose Condition \ref{cond:high1} holds. Let $\mathcal C_n=\{\theta\in\Theta:\tilde T_{n,R}(\theta)\le\tilde c_{n,R}(\theta)\}$.
	Then, for any $R\in\mathbb N$,
\begin{align}
	\liminf_{n\to\infty}\inf_{\mu\in\mathcal M}\inf_{(\theta,P)\in\mathcal F}P\Big(\theta\in\tilde{\mathcal C}_{n,R}\Big)\ge 1-\alpha.
\end{align}
\end{proposition}

\begin{proof}
Let $\theta\in \Theta_I(P).$	
Note that we have $\tilde T_{n,R}(\theta)=Z_{\mu,n}(\theta)+b_\mu(\theta),$ where $Z_{\mu,n}$ and $b_\mu$ are as in \eqref{eq:B1}-\eqref{eq:B2}.  Therefore, we may restate the coverage of $\theta$ by the confidence region as follows:
\begin{align}
\theta& \in \tilde{\mathcal C}_{n,R}\notag\\
&\Leftrightarrow \tilde T_{n,R}(\theta)\le \tilde c_{n,R}(\theta)\notag\\
&\Leftrightarrow  \tilde T_{n,R}(\theta)\le c_{n,R,1-\alpha}(\theta)+\sqrt n\mu\beta\notag\\
&\Leftrightarrow Z_{\mu,n,R}(\theta)+b_\mu(\theta) \le c_{n,R,1-\alpha}(\theta)+\sqrt n\mu\beta\notag\\
&\Leftarrow Z_{\mu,n,R}(\theta) \le c_{n,R,1-\alpha}(\theta),\label{eq:thm1_1}
\end{align}
where the last step follows from Definition \ref{def:musmooth} (i).

Let the asymptotic confidence size be $AsySz\equiv \liminf_{n\to\infty}\inf_{\mu\in\mathcal M}\inf_{(\theta,P)\in\mathcal F}P(\theta\in\tilde{\mathcal C}_n).$
Then, there is a sequence $\{n\}$ such that $(\mu_n,c_n)\in\mathcal M\times\mathcal C$, $ (\theta_n,P_n)\in\mathcal F$, and
$\liminf_{n\to\infty}P_n(\theta_n\in\tilde{\mathcal C}_n)=AsySz.$ By compactness of  $\Psi$ and $\mathcal M$, there is a further subsequence along which $\Omega_{P_n}(\theta_n)\to \Omega,$ and $\mu_n\to \mu$. For notational simplicity, we use the same index $\{n\}$ for the subsequence.
 By Condition \ref{cond:high1} (ii) and passing to a further subsequence, one has $c_{n,R,1-\alpha}(\theta_n){\to}c_{\mu,1-\alpha}$ almost surely. Again by Condition
\ref{cond:high1} (i), we then obtain
\begin{align}
P_n( Z_{\mu_n,n,R}(\theta_n) \le c_{n,R,1-\alpha}(\theta_n))\to Pr( Z_{\mu,R}\le c_{\mu,1-\alpha})\ge 1-\alpha. \label{eq:thm1_2}
\end{align}
By \eqref{eq:thm1_1} and \eqref{eq:thm1_2},
\begin{align}
\liminf_{n\to\infty}P_n(\theta_n\in\tilde{\mathcal C}_n)\ge  Pr( Z_{\mu,R}\le c_{\mu,1-\alpha})\ge 1-\alpha.
\end{align}
We  may therefore conclude $AsySz\ge 1-\alpha.$
\end{proof}

\subsection{Proof of Theorem 3.1}
\begin{proof}[\rm Proof of Theorem \ref{thm:size_control}]
Below, we show Condition \ref{cond:high1}.	

First, by Assumption \ref{as:onSmu} (ii),
\begin{align}
Z_{\mu,n,R}(\theta)&= S_\mu(\sqrt n\bar m_{n,R}(\theta),\hat \Sigma_{n,R}(\theta))-S_\mu(\sqrt nE_P[m(X_i,\theta)],\Sigma_P(\theta))\notag\\
&=\sqrt n\big(\phi_\mu(\hat V_{n,R}(\theta)^{-1/2}\bar m_{n,R}(\theta))-\phi_\mu(V_P(\theta)^{-1/2}E_P[m(X_i,\theta)])\big)\notag\\
&=D\phi_\mu[\tilde m_{n,R}(\theta)](\sqrt n(\hat V_{n,R}(\theta)^{-1/2}\bar m_{n,R}(\theta)-V_P(\theta)^{-1/2}E_P[m(X_i,\theta)])),\label{eq:sc1}
\end{align}
where $V_P=\text{diag}(\Sigma_P)$. The last equality in \eqref{eq:sc1} follows from the mean value theorem (applied componentwise), and $\tilde m_{n,R}(\theta)$ is a point between  $\hat V_{n,R}(\theta)^{-1/2}\bar m_{n,R}(\theta)$ and $V_P(\theta)^{-1/2}E_P[m(X_i,\theta)])$, which may be different across components.

Let $\Gamma = \text{cl}(\{V_{P}(\theta)^{-1/2}E_P[m(X,\theta)],(\theta,P)\in\mathcal F\})$. This is a compact subset of $\mathbb R^J$ due to Assumption \ref{as:onP} (vi).
Note that, there exists a $L>0$ such that, uniformly over $m\in\Gamma$,
\begin{align}
	D\phi_\mu[m]'\Omega_P(\theta)D\phi_\mu[m]\le \|\Omega_P\|_{op}\|D\phi_\mu[m]\|^2\le L.
\end{align}
where $\|A\|_{op}$ and $\|A\|_F$ denote the operator and Frobenius norms of a matrix $A$ respectively, and the last inequality follows from  $\|A\|_{op}\le \|A\|_F$, $\|\Omega_P(\theta)\|_F\le J$, and $\sup_{m\in\Gamma}\|D\phi_\mu[m]\|^2$ being uniformly bounded by the continuity of $D\phi_\mu[\cdot]$ (by Definition \ref{def:musmooth}) and compactness of $\Gamma$. As we show below, the asymptotic variance of $Z_{\mu,n,R}$ is then bounded by $L$. Below, we let $\bar c$ be the $1-\alpha$ quantile of a mean-zero normal distribution with variance $L$, which will serve as the uniform upper bound for the critical value.

Let $(\mu_n,c_n,\theta_n,P_n)\in [\underline M,\overline M]\times[0,\overline c]\times\mathcal F$ such that $\Omega_{P_n}(\theta_n)\to \Omega,$ and $(\mu_n,c_n)\to (\mu,c)\in [\underline M,\overline M]\times[0,\overline c].$
Then, by \eqref{eq:sc1},
\begin{align}
Z_{\mu_n,n,R}(\theta_n)&=D\phi_{\mu_n}[\tilde m_{n,R}(\theta_n)](\sqrt n(\hat V_{n,R}(\theta_n)^{-1/2}\bar m_{n,R}(\theta_n)-V_{P_n}(\theta_n)^{-1/2}E_{P_n}[m(X_i,\theta_n)]))\notag\\
&=D\phi_{\mu_n}[m^*](\sqrt n(\hat V_{n,R}(\theta_n)^{-1/2}\bar m_{n,R}(\theta_n)-V_{P_n}(\theta_n)^{-1/2}E_{P_n}[m(X_i,\theta_n)])+r_n,\label{eq:sc2}
\end{align}
where
\begin{align}
r_n=(D\phi_{\mu_n}[\tilde m_{n,R}(\theta_n)]-D\phi_\mu[m^*])(\sqrt n(\hat V_{n,R}(\theta_n)^{-1/2}\bar m_{n,R}(\theta_n)-V_{P_n}(\theta_n)^{-1/2}E_{P_n}[m(X_i,\theta_n)]).	\label{eq:sc3}
\end{align}
Note that, one may write
\begin{align}
	\bar m_{n,R}(\theta)=\frac{1}{n}\sum_{i=1}^n \hat m_R(X_i,u_{i},\theta),\label{eq:sc4}
\end{align}
where $\hat m_R(X_i,u_{i},\theta)=\frac{1}{R}\sum_{r=1}M(X_i,u_{i,r},\theta)$.

 By Assumption \ref{as:onP} (vi) and Minkowski's inequality, $E_P[|\hat m_R(X_i,u_{i},\theta)/\sigma_{P,j}(\theta)|^{2+\delta}]\le M'$ for some uniform constant  $0<M'<\infty$.
Together with Assumption \ref{as:onP} (iii),  we may then apply Lemma 2 in \cite{Andrews:2009aa}, which uses a triangular-array LLN and CLT to obtain
\begin{align}
\sqrt n(\hat V_{n,R}(\theta_n)^{-1/2}\bar m_{n,R}(\theta_n)-V_{P_n}(\theta_n)^{-1/2}E_{P_n}[m(X_i,\theta_n)]\stackrel{P_n}{\leadsto}\mathcal W,\label{eq:sc5}
\end{align}
where $\mathcal W$ is a multivariate normal random vector with covariance matrix $\Omega$.
By $\phi_\mu$ being a $\mu$-smooth approximation and Definition \ref{def:musmooth} (ii),
\begin{multline}
r_n
\le (K+\alpha/\mu_n)\times \|\tilde m_{n,R}(\theta_n)-m^*\|\times O_{P_n}(1)\\
\le (K+\alpha/\underline M)\times \|\tilde m_{n,R}(\theta_n)-m^*\|\times O_{P_n}(1)=o_{P_n}(1),\label{eq:sc6}
\end{multline}
where the last equality follows from  $\tilde m_{n,R}(\theta_n)$ being between $\bar m_{n,R}(\theta_n)$ and $m^*$ and the uniform law of large numbers applied to $\bar m_{n,R}(\theta_n)$, which again follows from Lemma 2 in \cite{Andrews:2009aa}. By Definition \ref{def:musmooth} (iii), Assumption \ref{as:onP} (vi), and the extended continuous mapping theorem \citep[][Theorem 1.11.1]{Van-Der-Vaart:1996aa},
we obtain
\begin{align}
Z_{\mu_n,n,R}(\theta_n)\stackrel{P_n}{\leadsto}D_\mu[m^*](\mathcal W)=:Z_{\mu,R}.\label{eq:sc7}
\end{align}
Note that $Z_{\mu,R}$ is a mean zero normal distribution  with variance $D_\mu[m^*]\Omega D_\mu[m^*]'$. By Assumption \ref{as:onP} (v) and (vii) and a continuity argument, $D_\mu[m^*]\ne 0$ and $\Omega$ is positive definite. Hence, the variance of $Z_{\mu,R}$ is bounded away from 0. This ensures that the cumulative distribution function of $Z_{\mu,R}$ is continuous and strictly increasing at its $1-\alpha$ quantile. This establishes Condition \ref{cond:high1} (i).

Next, we show Condition \ref{cond:high1} (ii).
Let a sequence $\{(\mu_n,\theta_n)\in[\underline M,\overline M]\times\Theta,n\ge 1\}$ be given, fix $c$, and let $\mathbf C$ be the set of sequences $\{P_n\}$ such that $\Omega_{P_n}(\theta_n)\to \Omega$ and $E_{P_n}[m(X,\theta_n)]\to m^*$ for some $\Omega\in\Psi$ and $m^*\in\mathbb R^J$. In what follows, we write $J_{n,R}(c,P_n)=P_n(Z_{\mu_n,n,R}(\theta_n)\le c)$ and $J(c)=Pr(Z_{\mu,R}\le c)$.
In Part (i), we have shown that
\begin{align}
	\rho_L(J_{n,R}(\cdot,P_n), J(\cdot))\to 0,~\text{as }n\to\infty,\label{eq:b2_1}
\end{align}
where $\rho_L$ is a metric that metrizes the weak convergence (e.g. the bounded Lipschitz metric).
For each $n,$ $\hat P_{n,R}$ be the empirical distributions of $\{X_{i},u_{i,1},\cdots,u_{i,R}\}_{i=1}^n$. Let $\hat P_{n,\infty}$ denote the joint distribution of $(X,u)$, where $X$'s law is the empirical distribution  based on the original sample $\{X_1,\cdots, X_n\}$, while $u$'s conditional law $P(\cdot|X)$  is known. The bootstrap sample $\{X_1^*,\cdots, X_n^*\}$ is drawn from the empirical distribution, while for each $X^*_i$, $\{u_{i,1},\cdots,u_{i,R}\}$ is drawn from $P(\cdot|X^*_i)$. Therefore,  $\{X^*_i,u_{i,1},\cdots,u_{i,R}\}_{i=1}^n$ can be viewed as a sample from $\hat P_{n,\infty}.$

Define
\begin{align}
	Z^*_{\mu_n,n,R,\infty}(\theta_n)&\equiv S_\mu(\sqrt n\bar m_{n,R}(\theta_n),\hat \Sigma_{n,R}(\theta_n))-S_\mu(\sqrt nE_{\hat P_{n,\infty}}[m(X_i,\theta_n)],\Sigma_{\hat P_{n,\infty}}(\theta_n)),\label{eq:b2_1a}
\end{align}
where for each $\theta\in\Theta$,
\begin{align}
	E_{\hat P_{n,\infty}}[m(X_i,\theta)]&=\bar m_n(\theta)=\frac{1}{n}\sum_{i=1}^n m(X_i,\theta)\\
	\Sigma_{\hat P_{n,\infty}}(\theta))&=\frac{1}{n}\sum_{i=1}^n(m(X_i,\theta)-\bar m_n(\theta))(m(X_i,\theta)-\bar m_n(\theta))'.
\end{align}
We then let $J_{n,R}(c,\hat P_{n,\infty})\equiv \hat P_{n,\infty}(Z^*_{\mu_n,n,R,\infty}\le c|X^n),$ where $X^n=(X_1,\cdots,X_n)$.
Note that $Z^*_{\mu_n,n,R,\infty}(\theta_n)$ differs from the root $Z^*_{\mu,n,R,R_2}$ we compute in Algorithm 1 as the centering term in \eqref{eq:b2_1a} is not based on any simulation (or it can be viewed as the limit with $R_2=\infty$).
Below, we first show that there is a subsequence $k_n$ such that $\{\hat P_{k_n,\infty}\}\in \mathbf C$ almost surely. We then apply Theorem 1.2.1 in \cite{Politis:1999aa} to establish bootstrap consistency for this infeasible bootstrap. Bootstrap consistency for Algorithm 1 is then established by showing that replacing $(E_{\hat P_{n,\infty}}[m(X_i,\theta)],\Sigma_{\hat P_{n,\infty}}(\theta)))$ with $(\bar m_{n,R_2},\hat\Sigma_{n,R_2})$ has asymptotically negligible effects as $R_2\to\infty.$

Under Assumption \ref{as:onP} and arguing as in \eqref{eq:sc5}, we may apply Lemma 2 in \cite{Andrews:2009aa}, which ensures that
\begin{align}
\bar m_{n,R}(X_i,\theta_n)-E_{P_n}[m(X_i,\theta_n)]=o_{P_n}(1),\label{eq:b2_2}
\end{align}
and
\begin{align}
	\hat V_{n,R}(\theta_n)^{-1/2}\hat \Sigma_{n,R}(\theta_n)\hat V_{n,R}(\theta_n)^{-1/2}\stackrel{P_n}{\to}\Omega.\label{eq:b2_3}
\end{align}

Hence, for any subsequence of these sequences that are converging in probability, one also has a further subsequence $\{k_n\}$ along which convergence in \eqref{eq:b2_2} and \eqref{eq:b2_3} hold almost surely under $P_n$ (instead of in probability) for all $n$. Therefore $\{\hat P_{k_n,\infty}\}\in\mathbf C$ with probability 1. By Theorem 1.2.1 in \cite{Politis:1999aa}, it then holds that $\rho_L(J_{k_n,R}(c,\hat P_{k_n,\infty}),J_{k_n,R}(c,P_{k_n}))\to 0$ with probability 1. Since the choice of the original subsequence was arbitrary, this in turn implies
\begin{align}
	\rho_L(J_{n,R}(\cdot,\hat P_{n,\infty}),J_{n,R}(\cdot,P_{n}))=o_{P_n}(1).\label{eq:b2_4}
\end{align}

Now consider the root $Z^*_{\mu,n,R,R_2}(\theta)=S_\mu(\sqrt n\bar m^*_{n,R}(\theta),\hat\Sigma^*_{n,R}(\theta))-S_\mu(\sqrt n\bar m_{n,R_2}(\theta),\hat \Sigma_{n,R_2}(\theta))$ we compute in Algorithm 1. For each $c\in\mathbb R$, let $L_{n,R_2}(c)\equiv \hat P_{n,\infty}(Z^*_{\mu_n,n,R,R_2}(\theta_n)\le c\mid X^n).$ Below we take $\rho_L$ to be the bounded Lipschitz metric. Then,
\begin{align}
	\rho_L(L_{n,R_2}(\cdot), J_{n,R}&(\cdot,\hat P_{n,\infty}))\notag\\
	&=\sup_{h\in \text{BL}_1}\Big|E_{\hat P_{n,\infty}}\Big[h(Z^*_{\mu,n,R,R_2}(\theta_n))\mid X^n\Big]-E_{\hat P_{n,\infty}}\Big[h(Z^*_{\mu,n,R,\infty}(\theta_n))\mid X^n\Big]\Big|\notag\\
	&=\Big|E_{\hat P_{n,\infty}}\Big[S_\mu(\sqrt n\bar m_{n,R_2}(\theta),\hat \Sigma_{n,R_2}(\theta))-S_\mu(\sqrt n\bar m_{n,\infty}(\theta),\hat \Sigma_{n,\infty}(\theta))\mid X^n\Big]\Big|.\label{eq:b2_5}
\end{align}
Arguing as in \eqref{eq:sc1}-\eqref{eq:sc2}, one has
\begin{align}
S_\mu(\sqrt n\bar m_{n,R_2}&(\theta),\hat \Sigma_{n,R_2}(\theta))-S_\mu(\sqrt n\bar m_{n,\infty}(\theta),\hat \Sigma_{n,\infty}(\theta))	\notag\\
&=D\phi_{\mu_n}[\tilde m_{n,R_2}(\theta_n)]\sqrt n(\hat V_{n,R_2}(\theta_n)^{-1/2}\bar m_{n,R_2}(\theta_n)-\hat V_{n}(\theta_n)^{-1/2}\bar m_{n}(\theta_n))=o_{\hat P_{n,\infty}}(1),\label{eq:b2_6}
\end{align}
where $\hat V_{n,R_2}(\theta)=diag(\Sigma_{\hat P_{n,\infty}}(\theta))$, and $\tilde m_{n,R_2}(\theta_n)$ is a point between $\hat V_{n,R_2}(\theta_n)^{-1/2}\bar m_{n,R_2}(\theta_n)$ and $\hat V_{n}(\theta_n)^{-1/2}\bar m_{n}(\theta_n))$.
The last equality follows by the assumption on $R_2$. Note that, conditional on $X^n$, $S_\mu(\sqrt n\bar m_{n,R_2}(\theta),\hat \Sigma_{n,R_2}(\theta))-S_\mu(\sqrt n\bar m_{n,\infty}(\theta),\hat \Sigma_{n,\infty}(\theta))$ is uniformly integrable, and hence it converges in $L^1$ by \eqref{eq:b2_6}. Therefore, by \eqref{eq:b2_5}, it follows that
\begin{align}
	\rho_L(L_{n,R_2}(\cdot), J_{n,R}&(\cdot,\hat P_{n,\infty}))=o_{P_n}(1).\label{eq:b2_7}
\end{align}
Combining \eqref{eq:b2_1}, \eqref{eq:b2_4} and \eqref{eq:b2_7}, we conclude that
\begin{align}
	\rho_L(L_{n,R_2}(\cdot),J(\cdot))=o_{P_n}(1).\label{eq:b2_8}
\end{align}
For any subsequence, one can then find a further subsequence for which \eqref{eq:b2_8} holds almost surely.
Mimicking the argument in the proof of Theorem 1.2.1 in \cite{Politis:1999aa} and $J(\cdot)$ being strictly increasing at its $1-\alpha$ quantile, the bootstrap critical value along the subsequence then satisfies
\begin{align}
c_{k_n,R,1-\alpha}(\theta_n)\stackrel{}{\to}c_{\mu,1-\alpha},
\end{align}
with probability 1. Since this holds for arbitrary subsequence, we have
\begin{align}
	c_{n,R,1-\alpha}(\theta_n)\stackrel{P_n}{\to}c_{\mu,1-\alpha},
\end{align}
for the original sequence. This establishes Condition \ref{cond:high1} (ii).

The conclusion of the theorem now follows by applying Proposition \ref{prop:size} with $\mathcal M=[\underline M,\overline M]$ and $\mathcal C=[0,\bar c]$.
\end{proof}

\section{$\mu$-smooth approximations of index functions}
We show below that the functions $S_\mu$ given in Table \ref{tab:musmooth} and footnote \ref{fnt:S2} are the $\mu$-smooth approximations of the corresponding index functions.

For $(i)$, $S(m,\Sigma)=\sum_{j=1}^{J}[m_{j}/\sigma_{j}]_{+}$ is
the sum of functions $S^{j}=\max\{m_{j}/\sigma_{j},0\},j=1,2,...J$.
For each $S^{j}$, it has a $\mu$-smooth approximation function $S_{\mu}^{j}=\mu\ln(\exp(\frac{m_{j}}{\mu\sigma_{j}})+1)$
with parameters $(1,\ln 2,0)$. By Lemma $2.1$ in \cite{Beck:2012xy}, $S(m,\Sigma)$ has the following $\mu$-smooth approximation
function
\begin{align}
   S_{u}(m,\Sigma) = u\sum_{j=1}^{J}\ln(\exp(\frac{m_{j}}{\mu\sigma_{j}})+1),
\end{align}
 with parameters $(J,J\ln(2),0)$.

For $(ii)$, $S(m,\Sigma)=\max_{j=1,..,J}\{m_{j}/\sigma_{j}\}_{+}=\max\{m_{1}/\sigma_{1},....,m_{J}/\sigma_{J},0\}$.
By Theorem $4.2$ in \cite{Beck:2012xy},  this suggests $S(m,\Sigma)$ has a $\mu$-smooth approximation function
$S_{\mu}(m,\Sigma)=\mu\ln(\sum_{j=1}^{J}\exp(\frac{m_{j}}{\mu\sigma_{j}})+1)$
with parameters $(1,\ln(J+1),0)$.

For the function in footnote \ref{fnt:S2}, $S_\mu$ is constructed through an operation called inf-convolution \citep[see e.g][]{Rockafellar:2009aa}.
Let $S(m,\Sigma)=\inf_{t\in\mathbb R^J_-}(m-t)'\Sigma^{-1}(t-m)$ and let $w:\mathbb R^J\to \mathbb R$ be defined by
\begin{align}
	w(m)\equiv \frac{m'\Sigma^{-1}m}{2}.\label{eq:ms5}
\end{align}
 Define the infimal convolution of $S$ and $w_\mu\equiv \mu w(\cdot/\mu)$ by
\begin{align}
	S_\mu(m,\Sigma)=\inf_{m'\in \mathbb R^J}\Big\{S(m')+\mu w(\frac{m-m'}{\mu})\Big\}.\label{eq:ms6}
\end{align}
This function serves as a $\mu$-smooth approximation of $S$. By Theorem 4.1 in \cite{Beck:2012xy}, $S_\mu$ then has the following dual formulation:
\begin{align}
	S_\mu(m,\Sigma)=\sup_{u\in\mathbb R^J}\Big\{m'u-S^*(u)-\mu w^*(u)\Big\}.
\end{align}
where $S^*$ and $w^*$ are the Fenchel conjugates of $S$ and $w$ respectively.

The Fenchel conjugate of $S$ is
\begin{align}
	S^*(u) = \sup_{m\in \mathbb R^J}m'u-\inf_{t\in\mathbb R^J_-}(m-t)'\Sigma^{-1}(t-m).\label{eq:ms1}
\end{align}
Hence, one has
\begin{align}
	-S^*(u)&=\inf_{m\in \mathbb R^J}\inf_{t\in\mathbb R^J_-}(m-t)'\Sigma^{-1}(m-t)-m'u=\inf_{t\in\mathbb R^J_-}\inf_{m\in \mathbb R^J}(m-t)'\Sigma^{-1}(m-t)-m'u.\label{eq:ms2}
\end{align}
Let $m^*$ be the optimal solution to the inner minimization on the right hand side of \eqref{eq:ms1}. It is a convex quadratic program whose first-order necessary condition is given as
\begin{align}
	2\Sigma^{-1}(m^*-t)=u,\label{eq:ms3}
\end{align}
which suggests $m^*=\Sigma u/2+t$. Substituting this into \eqref{eq:ms2}, we have
\begin{align}
	-S^*(u)=\inf_{t\in\mathbb R^J_-}\frac{u'\Sigma u}{4}-\frac{u'\Sigma u}{2}-t'u=-\frac{u'\Sigma u}{4}-\sup_{t\in\mathbb R^J_-}t'u=-\frac{u'\Sigma u}{4}-\delta_{\mathbb R^J_+}(u),\label{eq:ms4}
\end{align}
where $\delta_A(u)$ denotes the optimization indicator $\delta_A(u)=0$ if $u\in A$ and $\delta_A(u)=\infty$ otherwise. 
It is straightforward to show that its Fenchel conjugate is $w^*(u)=u'\Sigma u/2.$ Combining the results above, we obtain
\begin{align}
S_\mu(m,\Sigma)	=\sup_{u\in\mathbb R^J_+}\Big\{m'u-\frac{1+2\mu}{4}u'\Sigma u\Big\}.\label{eq:ms7}
\end{align}
Note that $w(0)=0$ and $\nabla_{m}w(m)=\Sigma^{-1}m$. Hence, one obtains  $\|\nabla_{m}w(m)-\nabla_{m}w(m^{'})\|=\|\Sigma^{-1}(m-m^{'})\|\leq\|\Sigma^{-1}\|\cdot\|m-m^{'}\|=(1/\sigma_{min})\cdot\|m-m^{'}\|$, where $\sigma_{min}$ is the smallest eigenvalue of $\Sigma$. By Corollary $4.1$ in \cite{Beck:2012xy}, it follows that $S_\mu(m,\Sigma)$ is a $\mu$-smooth approximation of $S$ with parameters $(1/\sigma_{min},D,0)$, where $D=\sup_{m}\sup_{r\in\nabla_{m}S}r^{'}\Sigma r/2$ and $D$ is assumed to be finite.
\end{document}